\documentclass{article} 

\usepackage{amsthm}
\newtheorem{proposition}{Proposition}
\newtheorem{corollary}{Corollary}
\newtheorem{theorem}{Theorem}
\newtheorem{lemma}{Lemma}
\newtheorem{example}{Example}
\usepackage{graphicx}
\usepackage{newtxtext,newtxmath}
\usepackage[usenames,dvipsnames]{xcolor}
\usepackage{amsfonts}
\usepackage{amsmath}

\usepackage{bbm}
\usepackage{subcaption}
\usepackage{pgfplots}
\usepackage{tikz}
\usetikzlibrary{patterns}

\newcommand{\R}{\mathbb{R}}

\usetikzlibrary{shapes,arrows}

\usepackage{enumitem}

\usepackage{subcaption}
\usepackage{booktabs}
\usepackage{float}

\usepackage{sidecap}

\usepackage[normalem]{ulem}

\usepackage{adjustbox}

\usetikzlibrary{shapes.misc}

\tikzset{cross/.style={cross out, draw=black, minimum size=2*(#1-\pgflinewidth), inner sep=0pt, outer sep=0pt},
cross/.default={6pt}}

\usepackage[hypertexnames=false]{hyperref}

\newcommand{\xmin}{\underline{x}}
\newcommand{\xmax}{\overline{x}}

\usepackage{mathtools}
\usepackage{multirow}

\begin{document}

\title{Function approximations for counterparty credit exposure calculations\thanks{This manuscript is currently under review. It incorporates revisions based on reviewer feedback
but has not yet been accepted for publication.}}

\author{Domagoj Demeterfi
\thanks{\texttt{d.demeterfi@qmul.ac.uk},
School of Mathematical Sciences, Queen Mary University of London, UK,
LAFRA, Faculty of Electrical Engineering and Computing, University of Zagreb, Croatia}  
   \and
        Kathrin Glau
        \thanks{\texttt{kathrin.glau@tum.de}, School of Mathematical Sciences, Queen Mary University of London, UK}		  
        \and
        Linus Wunderlich
        \thanks{\texttt{l.wunderlich@qmul.ac.uk}, School of Mathematical Sciences, Queen Mary University of London, UK}
}

\date{}

\maketitle

\begin{abstract}

The challenge to measure exposures regularly forces financial institutions into a choice between an overwhelming computational burden or oversimplification of risk. To resolve this unsettling dilemma, we systematically investigate replacing frequently called  derivative pricers by function approximations covering all practically relevant exposure measures, including quantiles. We prove error bounds for exposure measures in terms of the $L^p$ norm, $1 \leq p  < \infty$, and for the uniform norm. To fully exploit these results, we employ the Chebyshev interpolation and show exponential convergence of the resulting exposure calculations. As our main result we derive probabilistic and finite sample error bounds under mild conditions including the natural case of unbounded risk factors. We derive an asymptotic efficiency gain scaling with $n^{1/2-\varepsilon}$ for any $\varepsilon>0$ with $n$ the number of simulations. Our numerical experiments cover callable, barrier, stochastic volatility and jump features. Using 10\,000 simulations, we consistently observe significant run-time reductions in all cases with speed-up factors up to 230, and an average speed-up of 87. We also present an adaptive choice of the interpolation degree.  Finally, numerical examples relying on the approximation of Greeks highlight the merit of the method beyond the presented theory.

\end{abstract}

\section{Introduction}

The global financial crisis of 2007--2009 revealed dramatic consequences of miscalculations of risks associated to derivatives trading. The accurate estimation of counterparty exposures for complex financial derivatives poses a significant challenge to financial institutions. To capture market realities a derivative needs to be numerically valued for a large number of scenarios drawn from sophisticated stochastic models. This challenge typically results in an uncomfortable compromise between the model fit and computational as well as operational complexity.
We observe how severely that key challenge affects risk management practice, as we are aware of banks' intentions to discard internal models to fall back on standardized approaches in the course of the most recent regulatory developments \cite{wilkes2023}.
For example, even if a model as feasible as Heston's is employed by the pricing desk, the new challenges may well force the risk desk into the standardized approach based on the Black-Scholes-Merton (BSM) model delta and prudent adjustments. Such practice typically leads to higher capital requirements due to shortcomings of the BSM model that need to be compensated by simple adjustments. 
Academia can play a key role in order to facilitate the use of advanced stochastic models while maintaining practical feasibility of challenging risk computations. Overcoming this bottleneck allows industry to harvest decades of essential advancements in stochastic modelling for the critical task of responsible risk assessment. In order to move these developments forward we tackle the computational bottleneck of frequent pricer calls arising in the quantification of the market risk component of counterparty credit risk, i.e., in exposure calculations.

The counterparty credit \emph{exposure} $X_t$ is the amount that would be lost in the event that the counterparty to a derivative trade defaults at time $t$. For a non-collateralized trade, the exposure is simply the positive part of the stochastic value of the derivative $V_t$ depending on the future market state $Z_t$,
\begin{equation}\label{eq:exp_str}
X_t = V_t(Z_t)^+ = \max\left\{ V_t(Z_t), 0 \right\}.
\end{equation}
The future market state $Z_t \in \mathbb{R}^D$ comprising $D$ risk factors is a random vector on an atomless probability space $(\Omega, \mathcal{F}, \mathbb{P})$. Let the value function $V_t:\textrm{supp} (Z_t) \to \mathbb{R}$ be sufficiently regular so that $X_t \in L^0(\Omega, \mathcal{F}, \mathbb{P})$, where we follow the standard notation, as in, e.g., \cite[Appendix~A.7]{follmer2008}. {Additionally, for a function $f:\mathcal{D} \to \mathbb{R}$ and $S\subset \mathcal{D}$ we introduce the uniform norm as $\Vert f \Vert_{\infty,S} = \sup_{s \in S} |f(s)|$ since many numerical function approximations allow for error estimates using the genuine supremum instead of the essential one. Note that if $S=\mathcal{D}$ we drop the subscript $S$.} To measure the degree of exposure, we choose to define a (monetary) \emph{exposure measure} as a map $\rho : L^0 \to \mathbb{R}\cup \{\infty\}$ with only the following properties:
\begin{enumerate}
	\item $X_1 \leq X_2 \implies \rho(X_1) \leq \rho(X_2)$, $X_1, X_2 \in L^0$ (monotonicity),
	\item $\rho(X+c) = \rho(X) + c$, $X \in L^0, c \in \mathbb{R}$ (cash-additivity).
\end{enumerate}
We highlight that the definition does not assume sub-additivity in order to include all essential quantile-based exposure measures, such as the Potential Future Exposure, which is neither a coherent \cite{artzner1999} nor a convex \cite{follmer2002} risk measure. 
We note that vast majority of exposure measures encountered in practice is law-invariant, hence can be seen as maps of distribution functions, $\rho(X_t) = \rho(F_{X_t})$, see \cite[Remark~2.1]{weber2006}.
Let $Q_X(\cdot) = \inf	\left\{x \in \mathbb{R} : F_X(x) \geq \cdot \right\}$ be the quantile function associated to a random variable $X$ with the distribution function $F_X$. Examples of key exposure measures for confidence level $\alpha \in (0,1)$ are
    \begin{center}
    \begin{tabular}{ll}
        Potential Future Exposure (PFE) & $\operatorname{PFE}_\alpha(X) = Q_X(\alpha)$\\
        Credit Expected Shortfall (CES) & $\operatorname{CES}_\alpha(X) = \frac{1}{1-\alpha}\int_\alpha^1 Q_X(u) du$\\
        Expected Exposure (EE) & $\operatorname{EE}(X) = \mathbf{E}[X].$\\%\hline
    \end{tabular}
    \end{center}
    
Since the exact distribution is not known a priori, exposure estimation typically is based on Monte Carlo (MC) simulation, an approach known as \emph{full re-evaluation}. Therein, each derivative is re-priced  $x_t^i=V_t(z_t^i)^+$ for a large number $n\geq 1$ of simulated market scenarios $\{z_t^i\}_{i=1}^n$ and an \emph{empirical exposure estimator} $\widehat{\rho}(\mathbf{x}_t) = \rho(F_{\mathbf{x}_t})$ is computed from the empirical distribution $F_{\mathbf{x}_t}(x) = \frac{1}{n} \sum_{i=1}^n \mathbbm{1}_{\{x\geq x_t^i\}}$ with $\mathbf{x}_t = \{x_t^i\}_{i=1}^n$, see Appendix \ref{sec:fr}.

The highly flexible full re-evaluation method is conceptually simple, and lies on strong foundations backed by decades of extensive research in both academia and industry. This also means that a major financial institution already has the well-established simulation and pricing engines deployed in production. Unfortunately, the appealing properties of
a full re-evaluation come with prohibitive computational costs
as many hard-to-value financial derivatives
need to be considered simultaneously.

A variety of techniques have been proposed to accelerate the full re-evaluation. For example, \cite{glau2019fast,wunderlich2023,GRZELAK2022,karlsson2016,Schftner2008OnTE} employ function approximations techniques with promising results. Therefore, we believe it is worthwhile to gain a deeper understanding of accelerating full re-evaluation with a large class of function approximations.
In order to rigorously assess the effectiveness and distil preferable approximation methods, we formalize and analyse the approach in a generic framework. We deliberately only assume the availability of black-box evaluations of pricers, thus focusing on flexible acceleration methods with low implementation burden and hence high attractiveness for practice.

We propose an approach consisting of replacing the computationally expensive price evaluation $X_t=V_t(Z_t)$ with an efficient alternative $Y_t=U_t(Z_t)$ where $U_t \approx V_t$.
A typical choice for the function approximation is based on nodal evaluations of the reference pricer $V_t$. The advantage compared to full re-evaluation is that only very few nodal points $\{\zeta_t^j\}_{j=1}^N,$ $N \ll n $,
are required to specify the approximating function $U_t$ via approximation operator $\mathbf{I}$, namely, $U_t = \mathbf{I}(\{V_t(\zeta_t^j)\}_{j=1}^N)$. In Figure \ref{fig:process_flow} we give a graphical overview of the procedure.

\begin{figure}[htb!]
\begin{center}
\tikzstyle{block} = [rectangle, draw, fill=blue!20, 
    text width=6em, text centered, rounded corners, minimum height=4em]
\tikzstyle{line} = [draw, -latex']
\tikzstyle{solidline} = [draw, -latex']

\begin{tikzpicture}[node distance = 2cm, auto]
    \node [block] (simulation) {\textbf{Simulation}\\$z_t^i \sim Z_t$};
    \node [block, below of=simulation, node distance=2.5cm] (pricing) {\textbf{Pricing}\\$z_t \mapsto V_t(z_t) $};
    \node [block, below of=pricing, node distance=2.5cm] (measurement) {\textbf{Exposure measurement}};
    \node [block, right of=pricing, node distance=5.5cm, text width=10em, minimum height=5em] (approximation) {\textbf{Approximation}\\$U_t =\mathbf{I}(\{V_t(\zeta_t^j)\}_{j=1}^N) \approx V_t$\\$z_t \mapsto U_t(z_t)$};
    \path [draw, dashed, -latex'] (simulation) -- node[midway, left] {$\left\{z_t^i\right\}_{i=1}^n$} (pricing);
    \path [draw, dashed, -latex'] (pricing) -- node[midway, left] {$\mathbf{x}_t=\left\{V_t(z^i_t)^+\right\}_{i=1}^n$} (measurement);
    \path [line] (simulation) -- node[midway, above right] {$\left\{z_t^i\right\}_{i=1}^n$} (approximation);
    \path [line] (pricing) -- node[midway, above] {$\{V_t(\zeta_t^j)\}_{j=1}^N$} (approximation);
    \path [line, line width=0.1mm] (pricing) -- node[midway, below] {$N \ll n$} (approximation);
    \path [line] (approximation) -- node[midway, below right] {$\mathbf{y}_t=\left\{U_t(z^i_t)^+\right\}_{i=1}^n$} (measurement);
    \draw [dashed, -latex'] (measurement) -- +(0,-1.5) node[below] {$\widehat{\rho}(\mathbf{x_t})\approx \rho(V_t(Z_t)^+)$};
    \draw [solidline] (measurement) -- +(4,-1.5) node[below] {$\widehat{\rho}(\mathbf{y_t})\approx \rho(U_t(Z_t)^+)$};
\end{tikzpicture}
\end{center}
\vspace{-0.5cm}
\caption{The scheme shows how function approximation is employed to accelerate exposure measurement. Dashed arrows represent the standard full re-evaluation, while solid arrows highlight the steps of the accelerated calculation. }
    \label{fig:process_flow}
\end{figure}
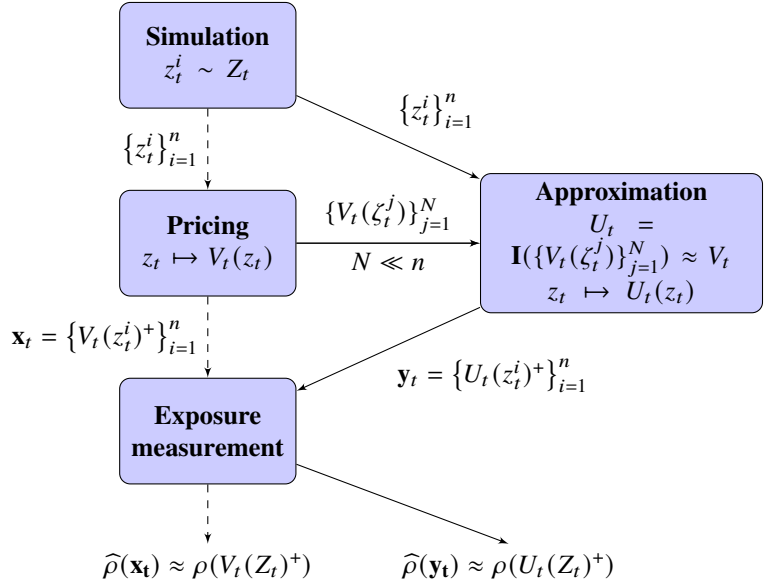

Key to our analysis is assessing the total simulation error $|\rho(X_t)-\hat{\rho}(\mathbf{y}_t)|$ where $\mathbf{y}_t$ is the data sample obtained analogously to $\mathbf{x}_t$ when $V_t$ is replaced by $U_t$. First note
\begin{equation}\label{eq:err_split_new}
\left|\rho(X_t) -  \widehat{\rho}(\mathbf{y}_t) \right|  \leq \left|\rho(X_t) -  \widehat{\rho}(\mathbf{x}_t) \right| + \left|\widehat{\rho}(\mathbf{x}_t) - \widehat{\rho}(\mathbf{y}_t)\right| .
\end{equation} 
The first term is the standard simulation error, which is equal to the error of the full re-evaluation and can be expected to decay with $1/\sqrt{n}$, which we verify for EE, PFE and CES in Appendix \ref{sec:exp_mea_exa}. Considering the second term, for any law-invariant exposure measure $\rho$, and for bounded risk factors $Z_t$, and $V_t$ and $U_t$ sufficiently regular so that $X_t,Y_t\in L^\infty$, we have
\begin{equation}
    \left|\widehat{\rho}(\mathbf{x}_t) - \widehat{\rho}(\mathbf{y}_t)\right| = \left|\rho(F_{\mathbf{x}_t}) - \rho(F_{\mathbf{y}_t})\right| \leq \Vert Q_{\mathbf{x}_t} - Q_{\mathbf{y}_t} \Vert_\infty
\end{equation}
thanks to the 1-Lipschitzness of law-invariant exposure measures as maps on the space $\mathcal{M}_{1,c}(\mathbb{R})$ of the probability measures compactly supported in $\mathbb{R}$ with respect to the $\infty$-Wasserstein metric $W_\infty(F_X,F_Y) = \sup_{0<u<1}|Q_X(u)-Q_Y(u)|$, as shown in \cite[Lemma~2.1,~Remark~2.2]{weber2004} together with \cite[Lemma~2.3.2]{weber2006} based on \cite[Lemma~4.2]{follmer2002}.

Inserting $Q_{\mathbf{x}}(u) = x^{\left(\lfloor nu \rfloor + 1\right)}$, where $x^{(k)}$ is $k$-th smallest element of the set $\mathbf{x}=\{x^i\}_{i=1}^n$ yields $\Vert Q_{\mathbf{x}_t} - Q_{\mathbf{y}_t} \Vert_\infty = \sup_{0<u<1}|x_t^{(\lfloor nu \rfloor + 1)} - y_t^{(\lfloor nu \rfloor + 1)}|$.
Then, employing Lemma \ref{lem:ord_diffs},
we obtain
    \begin{equation*}
        \Vert Q_{\mathbf{x}_t} - Q_{\mathbf{y}_t} \Vert_\infty = \max_{i=1,...,n} \left|x_t^{(i)}-y_t^{(i)}\right| \leq \Vert \mathbf{x}_t - \mathbf{y}_t \Vert_\infty \leq \Vert V_t^+ - U_t^+ \Vert_{\infty} \leq \Vert V_t - U_t \Vert_{\infty}.
    \end{equation*}
Thus we have $\left|\widehat{\rho}(\mathbf{x}_t) - \widehat{\rho}(\mathbf{y}_t)\right| \leq \Vert V_t - U_t \Vert_{\infty}$, and hence the following proposition.
\begin{proposition}\label{prop:uniform_bound_11}
    Let $Z_t$ be risk factors with compactly supported distribution and $V_t$ sufficiently regular so that $X_t\in L^\infty$. For any law-invariant exposure measure $\rho$ on $\mathcal{M}_{1,c}(\mathbb{R})$,  we have
    \begin{equation}\label{eq:err_split}
\left|\rho(X_t) -  \widehat{\rho}(\mathbf{y}_t) \right|  \leq \left|\rho(X_t) -  \widehat{\rho}(\mathbf{x}_t) \right| + \Vert V_t - U_t \Vert_{\infty}.
\end{equation} 
\end{proposition}

Next, we explore the analogue of the error splitting \eqref{eq:err_split_new}, namely,
\begin{equation}\label{eq:total_err_approx_new}
\left|\rho(X_t) -  \widehat{\rho}(\mathbf{y}_t) \right| \leq | \rho(X_t) - \rho(Y_t) | + | \rho(Y_t) - \widehat{\rho}(\mathbf{y}_t) |.
\end{equation}
By definition, any exposure measure $\rho$ is monotone and cash-additive and thus 1-Lipschitz with respect to the uniform norm on the subspace of bounded random variables, see \cite[Lemma~4.3]{follmer2008}. To see that this holds more generally, only assuming $\Vert X - Y \Vert_{\infty}<\infty$, first note that $X \leq Y + \Vert X - Y \Vert_{\infty}$ for $X, Y \in L^0$. Now, thanks to the monotonicity and cash-additivity we obtain $\rho(X) \leq \rho(Y) + \Vert X - Y \Vert_{\infty}$. The analogous calculation with the reversed roles of $X$ and $Y$ yields
\begin{equation*}
|\rho(X) - \rho(Y)| \leq \Vert X - Y \Vert_{\infty}.
\end{equation*}
We obtain $\Vert X_t - Y_t \Vert_{\infty}  =  \sup_{\omega \in \Omega} | V_t(Z_t(\omega)) - U_t(Z_t(\omega)) |\leq \Vert U_t - V_t \Vert_{\infty}$. If $\Vert U_t - V_t \Vert_{\infty} < \infty$, we have $|\rho(X_t) - \rho(Y_t) |_{\infty} \leq \Vert V_t - U_t \Vert_{\infty}$, and hence the following proposition, whose claim is trivial if $\Vert U_t - V_t \Vert_{\infty} = \infty$.
\begin{proposition}\label{prop:uniform_bound_12}
For exposures $X_t = V_t(Z_t)^+$ and $Y_t = U_t(Z_t)^+$, and any exposure measure $\rho$ we have
\begin{equation}\label{eq:err_split123}
\left|\rho(X_t) -  \widehat{\rho}(\mathbf{y}_t) \right|  \leq  \Vert V_t - U_t \Vert_{\infty}+\left|\rho(Y_t) -  \widehat{\rho}(\mathbf{y}_t) \right|.
\end{equation}
\end{proposition}

Propositions \ref{prop:uniform_bound_11} and \ref{prop:uniform_bound_12} showcase decompositions of the overall error into a MC error and the acceleration error due to function approximation, which then can be analysed further independently. Decomposition \eqref{eq:err_split} is based on the control of the MC-error of full re-evaluation. 
The advantage of decomposition \eqref{eq:err_split123} is that we can estimate the MC error purely based on $Y_t = U_t(Z_t)^+$,  where $U_t$ is chosen to be quick to evaluate. This can serve as a basis for adaptive algorithms optimally choosing the function approximation accuracy and the number of simulations.

Notice that in both propositions the acceleration error is bounded by the function approximation error in the uniform norm.  In order to exploit this structural insight, a highly efficient approximation of the value function in the uniform is required.
One particular method stands out for this task: Chebyshev interpolation, which is designed to yield a highly efficient approximation in the uniform norm that is strikingly fast to construct and to evaluate. It has already been shown to exhibit excellent approximation properties in theory and practice for option pricing since price functions are typically differentiable or even analytic and Chebyshev interpolation converges uniformly of polynomial order for differentiable functions, and exponentially fast for analytic functions.

Let $I_N$ denote Chebyshev interpolation of degree $N\geq 1$ and $\mathcal{E}(\left[a, b\right], \varrho)$ the Bernstein ellipse with parameter $\varrho > 1$ (see Appendix \ref{sec:cheby_acc}).
Assuming the simulation-based approach allows for a bounded interpolation domains we obtain:

\begin{corollary}\label{cor:cheb_conv}
Let $V_t : \left[a,b\right] \to \mathbb{R}$ be a derivative value function and denote by $U_t = I_N(V_t)$ its Chebyshev interpolation of degree $N$, and $\mathbf{x}_t$ and $\mathbf{y}_t$ the realized exposures induced by $V_t$ and $U_t$. Let $\rho$ be a law-invariant exposure measure.
\begin{enumerate}
	\item If $V_t$ can be analytically extended to a generalized Bernstein ellipse $\mathcal{E}_{\varrho} = \mathcal{E}([a, b], \varrho)$ for $\varrho>1$, then for  $C =\sup_{z \in \mathcal{E}_{\varrho}} |V_t(z)|$
	\begin{equation}\label{eq:cor_err_bound}
	\left|\widehat{\rho}(\mathbf{x}_t) -  \widehat{\rho}(\mathbf{y}_t) \right| \leq 4C\frac{\varrho^{-N}}{\varrho-1}.
	\end{equation}
	\item If $V_t$ and its derivatives up to $V_t^{(p-1)}$ are absolutely continuous and $V_t^{(p)}$ is of bounded variation $C$, for $1 \leq p < N$, then
	\begin{equation*}
	\left|\widehat{\rho}(\mathbf{x}_t) -  \widehat{\rho}(\mathbf{y}_t) \right| \leq \frac{4C}{\pi p} \left(\frac{b - a}{2(N-p)}\right)^p.
	\end{equation*}
\end{enumerate}
\end{corollary}
\begin{proof}
Theorems 8.2 and 7.2 in \cite{trefethen2013} show the decay of the error of the interpolation on $[-1, 1]$ in the uniform norm. The claims follow by applying these theorems to $V_t \circ \tau$ for $\tau$ as in \eqref{eq:domain_transform} and then using the inequality before Proposition \ref{prop:uniform_bound_11}. See Appendices B and C in \cite{gass2018} to follow how the transformation $\tau$ modifies the constants.

\end{proof}

In the setting of Corollary \ref{cor:cheb_conv} also the derivatives of $V_t$, known as Greeks, are approximated well by the derivatives of the Chebyshev interpolant, see \cite[Theorem~21.1]{trefethen2013} for the analytic, and \cite[Lemma~3.1]{Canuto1982ApproximationRF} for the differentiable case. We demonstrate the value of fast and accurate Greeks in counterparty risk applications in Section \ref{sec:cheby_greeks}.

Proofs of analyticity of option prices are established in \cite{gass2018}, covering a variety of payoff profiles, such as those of European, digital, and barrier options, in different models, including the BSM, MJD, and HSV models; see also \cite{gass2016,glau2019fast,potz2020}. The second part of Corollary \ref{cor:cheb_conv} is still applicable when pricing functions are only differentiable up to some order. For instance, this it the case for American options in the BSM \cite{Peskir2006OptimalSA} and MJD \cite{bayraktar2009} models. The availability of the option's delta is typically  a requirement for a useful model, which shows that differentiability is a very mild assumption.

In the remainder of the article, we tackle the following questions.
1) Going beyond the given setting, a natural question is how to accommodate unbounded risk factors. The uniform norm of Chebyshev approximations is typically controlled on compact domains. The same holds true for, e.g., universal approximations by NNs, but typically the error outside the domain cannot be controlled in a uniform way. 2) To deepen the understanding of efficiency, we need to assess the interplay between the number of simulations, the approximation degree and the overall error. 3) Function approximations in the uniform norm prove highly suitable to accelerate exposure calculations. Besides possible multivariate extensions by highly efficient complexity reduction methods, $L^p$-based approximations could be of interest as well. This calls for a preliminary error analysis in terms of the $L^p$ norm. 4) To assess the main approach in terms of essential practical requirements, efficient algorithms and solid numerical experiments are required.

\subsection{Main contributions and outline of the paper}

In the next section, as our main theoretical result we derive a probabilistic finite-sample error bound under mild conditions that are fulfilled for Chebyshev interpolation, answering questions 1) and 2).
To address question 3), in Section \ref{sec:lp_analysis} we  analyse the acceleration of exposure calculations using $L^p$ approximations. We derive error bounds and discuss them in detail revealing the dependence of the error on the distribution of the risk factors and the structure of the exposure measure. Regarding question 4), in Section \ref{sec:num_exp} we assess the method numerically for models and derivative products of varying levels of sophistication, and a range of different exposure measure and counterparty risk tasks. We employ suitable reference pricers and assess speed-ups and maximum relative errors. In the most challenging case considered, encompassing early-exercise and stochastic volatility features, replacing the reference finite-difference pricer reduces the computational time from 11.5 hours to 15 minutes, with only small relative errors below $1.2\%$. The example indicates that the proposed acceleration method may well contribute to avoid oversimplification of risk computations. The related literature is discussed in Section \ref{sec:literature}. Section \ref{sec:conc} on the conclusion and Section \ref{sec:outlook} on the outlook complement the main part of the article. The appendix details the full re-evaluation and the proposed acceleration method, discusses examples of empirical estimators, provides deferred proofs as well as algorithms.

\section{Main result on the exposure acceleration error for uniform approximations}\label{sec:main_result}

We turn to the case where the uniform function approximation error is only controlled on a subdomain of the support of the risk factors.
This occurs naturally when computational methods based on finite domains, such as finite difference solvers, are employed, whereas the support of the risk factors is unbounded.
Here, we need to account for the effect of the truncation on exposure estimates. We approach this by deriving a probabilistic error bound, which allows us to leverage the probability of risk factors to lie inside the truncated domain. To this end, we first provide a suitable setting.

Let $n$ be the number of exposure samples used in the full re-evaluation and $\Omega_L\subset\mathbb{R}^D$  a finite approximation domain parametrized by a size parameter $L$. We construct the function approximation $U_t = \mathbf{I}_N V_t^M\approx V_t$ of degree $N$ on $\Omega_L$ based on $N$ function evaluations of $V_t^M$ at nodal points $\{\zeta_t^j\}_{j = 1}^N \subset \Omega_L$, where $V_{t}^{M}$ is a Monte-Carlo pricer  with $M$ paths. For $n$ and $M$ sufficiently large,
\begin{equation}
\rho(X_t) - \widehat{\rho}(\mathbf{X}_t)\sim \mathcal{N}\left(0, \frac{\sigma_{\rho}^2}{n} \right) \hspace{0.05cm}\text{and}\hspace{0.05cm}\ V_t(z) - V_{t}^{M}(z) \sim \mathcal{N}\left(0,\frac{\sigma_z^2}{M}\right), \sigma_{z} \leq \overline{\sigma} \tag{A1}\label{ass0}
\end{equation}
can be reasonably assumed, which we do. We also assume
\begin{equation}
\hspace{0.83cm}\mathbb{P}(Z_t \notin \Omega_L) \leq \alpha e^{-\beta L^{\gamma}} \hspace{0.15cm}\text{with}\hspace{0.15cm} \alpha, \beta, \gamma>0 \tag{A2}\label{ass1}
\end{equation}
on the distribution of the risk factors, and
\begin{equation}
\hspace{0.83cm}\Vert V_t - \mathbf{I}_N V_t \Vert_{\infty,\Omega_L}\leq a e^{b(L^\vartheta  - \sqrt[D]{N})} \hspace{0.15cm}\text{with}\hspace{0.15cm} a,b>0, \vartheta \geq 2 \tag{A3}\label{ass2}
\end{equation}
for the approximation scheme's convergence, with the following stability condition
\begin{equation}
\hspace{0.83cm}\Vert \mathbf{I}_N V_t - \mathbf{I}_N V_{t}^M \Vert_{\infty,\Omega_L} \leq \mathfrak{c}(1+\ln N)^D \max_{j = 1,..., N} |V_t(\zeta_t^j) - V_{t}^M(\zeta_t^j)|  \tag{A4}\label{ass3}
\end{equation}
with $\mathfrak{c}>0$. Lastly, we make the following
\paragraph{Assumption \normalfont{(A5)}.} The computational effort to construct the approximation $\mathbf{I}_N V_t^M$ scales linearly  with $N^\xi$ for some $\xi>0$, the effort to evaluate it at one point scales linearly with $N$ and the effort to calculate $\widehat{\rho}(\mathbf{x})$ for $\mathbf{x} = (x^i)_{i=1}^n$ with $n\ln n$.

\begin{theorem}\label{thm:probErrBnd}
Under assumptions (A1)-(A5), for any $\kappa >0$, we can choose the approximation degree $N$, the domain size $L$, and the number of pricing paths $M$, such that
\begin{equation}\label{eq:peb}
\mathbb{P}\left(\left\vert \rho(X_t) - \widehat{\rho}(\mathbf{Y}_t)\right\vert \leq \frac{\kappa}{\sqrt{n}} \right) \geq 1-3e^{-c\kappa^2}
\end{equation}
for some $c>0$, and the computational effort to construct the function approximation and evaluate $\widehat{\rho}(\mathbf{Y}_t)$ scales linearly with a factor $\Lambda(n) \leq n^{1+\varepsilon}$ for any $\varepsilon>0$ and sufficiently large $n$. 
\end{theorem}

\begin{proof}
In Appendix \ref{app:probErrBoundProof}.
\end{proof}

Theorem \ref{thm:probErrBnd} enables us to discuss the speed-up of the proposed method compared to the standard full re-evaluation.
To this aim, we assess the standard full re-valuation in equivalent terms.
Ignoring the numerical pricer error, (A1) yields $\mathbb{P}\left(\left\vert\rho(X_t)- \widehat{\rho}(\mathbf{X}_t)\right\vert \leq \kappa/\sqrt{n}\right)\geq 1 - e^{-c\kappa^2}$ with the same constant $c$ as in \eqref{eq:peb}.
To quantify the computational effort, we choose $M = \sqrt{n}$ as the fixed number of pricing paths which optimally balances the computational cost of the inner and outer simulations, see \cite{broadie2011,Gordy2010}. This results in the computational cost of the full re-evaluation being proportional to $n\sqrt{n}$. Thus, for any $\varepsilon>0$, the acceleration achieves an asymptotic speed-up factor of
$ n^{1/2-\varepsilon}$.

It is interesting to note that this assertion holds true for any dimension $D \geq 1$. This might be surprising since we allow for function approximations that suffer from the curse of dimensionality as in (A3). For further insight, from the proof we see that $\Lambda(n) = n(\ln{n})^{D\max\{1, 2/\gamma\}}(\ln^2{n})^{2D+1} $. While the term is exponential in $D$ for a fixed~$n$, the effort is dominated by $n^{1+\varepsilon}$ for $D$ fixed and $n$ large enough, for any $\varepsilon>0$, i.e., arbitrarily close to only a linear growth.

While (A1) is a standard assumption on MC simulation, let us discuss the scope of assumptions (A2)-(A5) in more detail.

Assumption (A2) on tail probabilities is fulfilled in the uni- and multivariate BSM models for $\gamma = 2$ on a hyperrectangle $\Omega_L$ with the length of the largest side equal to $L$. As $\gamma$ can be chosen smaller, models with heavier tails are included as well. 

Pricing functions are analytic for a large set of models and options in which case they allow for an approximation with convergence exponential in $N^{1/D}$ as in (A3), see \cite{gass2016,gass2018}. Prominent examples of such approximations are the uni- and multivariate Chebyshev interpolation and the deep neural network approximation from \cite{Herrmann2022} which also fulfil the stability condition (A4) and the assumption on the computational effort (A5). For these methods, the assumption on the entailed dependence on the domain size as in (A3) holds for a wide class of options and models, see Appendix \ref{sec:ass_a3}.

Finally, (A5) assumes that the evaluation of $\widehat{\rho}(\mathbf{x})$ given $\mathbf{x} = (x^i)_{i=1}^n$ at most scales linearly with $n\ln n$, which is true for all quantile-based exposure measures, see the discussion following equation \eqref{eq:rho_m_emp_est}.

\section{Least-squares and $L^p$ approximations}\label{sec:lp_analysis}

To facilitate the analysis we observe that the following class of law-invariant exposure measures is monotone and cash-additive,
\begin{equation}\label{eq:qbrm_new}
\rho_m(X) = \int_0^1 Q_X(u)\,m(du),
\end{equation}
where \(m\) is a probability measure on \((0,1)\) with the density $f_m$. The class includes a large family of exposure measures for counterparty credit risk management, see Appendix \ref{sec:exp_mea_exa} for more details.

We have
\begin{equation}\label{eq:qbrm_bound}
|\rho_m(X_t) - \rho_m(Y_t)| \leq \int_0^1 |Q_{X_t}(u)-Q_{Y_t}(u)|m(du),
\end{equation}
and see that the error induced by the acceleration depends on the closeness of the quantile functions and the probability measure $m$ of the exposure measure $\rho_m$.
By the H\"{o}lder inequality, for $1 \leq r, r' \leq \infty$, $\frac{1}{r}+ \frac{1}{r'} = 1 $ we have the following bound
\begin{equation}\label{eq:qbrm_bound_holder_new}
|\rho_m(X_t) - \rho_m(Y_t)| \leq \Vert Q_{X_t} - Q_{Y_t}\Vert_{L^r(0,1)} \cdot \Vert f_m \Vert_{L^{r'}(0,1)}.
\end{equation}
To connect the difference of two quantile functions with the function approximation error, we employ an appropriate characterization of $p$-Wasserstein distance and so obtain the following proposition.

\begin{proposition}\label{prop:lp_bound}
Let $X_t = V_t(Z_t)^+$ and $Y_t = U_t(Z_t)^+$ be the exposures induced by the value function $V_t$, driven by the risk factors $Z_t$, and its approximation $U_t \approx V_t$. For quantile-based exposure measures $\rho_m$ as in \eqref{eq:qbrm_new} we have
\begin{equation}\label{eq:lp_bound}
		|\rho_m(X_t) - \rho_m(Y_t)| \leq \Vert V_t - U_t \Vert_{L^p(\mathbb{R}^D)} \cdot \Vert f_{Z_t} \Vert_{L^{q'}(\mathbb{R}^D)}^{\frac{1}{r}}\cdot \Vert f_m \Vert_{L^{r'}(0,1)},
		\end{equation}
	for any $p = rq$, $1 \leq r,q <\infty$, $\frac{1}{r}+\frac{1}{r'}=1$, $\frac{1}{q}+\frac{1}{q'}=1$, assuming $Z_t$ admits density $f_{Z_t}$.
\end{proposition}

\begin{proof}
We observe that 
\begin{equation*}\label{eq:q_lp_bound}
\Vert Q_{X_t} - Q_{Y_t}\Vert_{L^r(0,1)}^r \leq \Vert X_t - Y_t \Vert_{L^r(\Omega, \mathbb{P})}^r = \int_{\mathbb{R}^D} |V_t(z)^+-U_t(z)^+|^r f_{Z_t}(z)dz
\end{equation*}
directly follows from the characterization of $r$-Wasserstein distance
\begin{equation*}
W_r(X,Y)^r = \inf \left\{ \Vert X'-Y' \Vert_{L^r(\Omega)} | F_{X'} = F_X, F_{Y'} = F_Y \right\} = \int_0^1 |Q_X(u) - Q_Y(u)|^r du
\end{equation*}
for $1 \leq r < \infty$, which is derived from  \cite[Corollary~7.4.6]{Rachev2013TheMO}. Using the H\"{o}lder inequality for $1 \leq q \leq\infty$, $\frac{1}{q}+\frac{1}{q'}=1$, we obtain
\begin{equation*}\label{eq:qfs_lp_bound}
\begin{split}
\Vert Q_{X_t} - Q_{Y_t}\Vert_{L^r(0,1)}^r &\leq \Vert |V_t^+ - U_t^+|^r \Vert_{L^{q}(\mathbb{R}^D)} \cdot \Vert f_{Z_t} \Vert_{L^{q'}(\mathbb{R}^D)} \\
&= \Vert V_t^+ - U_t^+ \Vert_{L^{p}(\mathbb{R}^D)}^r \cdot \Vert f_{Z_t} \Vert_{L^{q'}(\mathbb{R}^D)}\\
&\leq \Vert V_t - U_t \Vert_{L^{p}(\mathbb{R}^D)}^r \cdot \Vert f_{Z_t} \Vert_{L^{q'}(\mathbb{R}^D)},
\end{split}
\end{equation*}
for $p = rq$.
Combining the above bound with \eqref{eq:qbrm_bound_holder_new} completes the proof.

\end{proof}

%\begin{remark}\label{rem:lst_sq}
While Proposition~\ref{prop:lp_bound} is formulated for general $p$, the least-squares fit of pricers corresponding to $p=rq=2$ is of greatest practical interest. For the particular choice $r=1$ and $q=2$ \eqref{eq:lp_bound} reads
\begin{equation*}
		|\rho_m(X_t) - \rho_m(Y_t)| \leq \Vert V_t - U_t \Vert_{L^2(\mathbb{R}^D)} \cdot \Vert f_{Z_t} \Vert_{L^{2}(\mathbb{R}^D)}\cdot \Vert f_m \Vert_{\infty},
\end{equation*}
while for $r=2$ and $q=1$ we have
\begin{equation*}
		|\rho_m(X_t) - \rho_m(Y_t)| \leq \Vert V_t - U_t \Vert_{L^2(\mathbb{R}^D)} \cdot \Vert f_{Z_t} \Vert_{\infty}^{\frac{1}{2}} \cdot \Vert f_m \Vert_{L^{2}(0,1)}.
\end{equation*}
 
\subsection{Comparison to uniform approximations in a tractable special case}\label{sec:example_new}

Like the uniform estimates, the final exposure error depends on the quality of the underlying function approximation. In addition, the $L^p$ error estimate for $1\leq p < \infty$ also depends on: 1) the structure of each exposure measure under consideration; and 2) the distribution of underlying risk factors.

These three components present in the error estimate are indeed driving the error itself which becomes apparent in the special case of positive, increasing and continuous $V_t$ and $U_t$ where
\begin{equation*}
\rho_m(X_t) - \rho_m(Y_t) = \int_{-\infty}^{\infty} \left(V_t(z)-U_t(z)\right) f_{Z_t}(z) f_m\left(F_{Z_t}(z)\right)dz
\end{equation*}
under the assumptions of Proposition \ref{prop:lp_bound}.
This special case is illustrative as it circumvents the key technical difficulty of the more general case. The representation reveals how the exposure error relates to the function approximation.
Namely, we see that $f_{Z_t}$ and $f_m\circ F_{Z_t}$ highly weighing large point-wise  errors $V_t(z) - U_t(z)$ can cause an error in the exposure measurement larger than  $\Vert V_t - U_t \Vert_{L^p(\mathbb{R})}$. However, the error estimate \eqref{eq:lp_bound} shows that it is smaller than a constant times this norm and thus that smaller function approximation errors lead to smaller exposure errors. While for uniform approximations this constant would be one, we next construct an example where the constant in the $L^p$ estimate is larger.

In this example, we consider value functions $V_t$ and $U_t$ of two digital options driven by a normally distributed risk factor $Z_t$ yielding tractable exposure calculations. Here, we choose $U_t$ close to the  value function  $V_t$ in the $L^2$ norm so that it can play the role of the approximation of $V_t$ in Scheme \ref{fig:process_flow}.

Specifically, we choose
$Z_t\sim\mathcal{N}(-\frac{1}{2}\sigma^2 t, \sigma^2 t) \hspace{0.1cm} \text{with} \hspace{0.1cm} t = 1/24 $ and $ \sigma = 0.05,$
$V_t(z) = \Phi\left(d^{k_1}(z)\right)$ with $ k_1 = 0.03,$ and $ U_t(z) = \Phi\left(d^{k_2}(z)\right)$ with $ k_2 = 0.04,$
where  $d^k(z) = (z - k -\frac{1}{2}\sigma^2\tau)/{ \sqrt{\sigma^2\tau}}$ with $ \tau = 1/24,$ and $\Phi$ the CDF of a standard normal distribution. Then $X_t = V_t(Z_t)$ and $Y_t = U_t(Z_t)$ correspond to the exposures of two digital call options with different log-strikes but same otherwise, see Figure \ref{fig:valuefs_digi}.

The $L^2$ distance between these value functions is $\Vert V_t - U_t \Vert_{L^2(\mathbb{R})} = 0.0516$.
In Figure \ref{fig:temp2}, we observe the following exposure measurements
\begin{equation*}
\begin{split}
\mathrm{PFE}_{0.99}(X_t) = 0.2666,&\hspace{0.5cm} \mathrm{CES}_{0.99}(X_t) = 0.3904,\\
\mathrm{PFE}_{0.99}(Y_t) = 0.0545,&\hspace{0.5cm} \mathrm{CES}_{0.99}(Y_t) = 0.1139,
\end{split}
\end{equation*}
and note that
\begin{equation*}
\begin{split}
|\mathrm{PFE}_{0.99}(X_t) - \mathrm{PFE}_{0.99}(Y_t)| > 4\cdot\Vert V_t - U_t \Vert_{L^2(\mathbb{R})},\\
|\mathrm{CES}_{0.99}(X_t) - \mathrm{CES}_{0.99}(Y_t)| > 5\cdot\Vert V_t - U_t \Vert_{L^2(\mathbb{R})}.
\end{split}
\end{equation*}
In other words, the exposure measurements induced by $V_t$ and $U_t$ differ significantly more than their $L^2$ distance.

The example highlights the previosly discussed influence of different factors on the final error through a practically relevant example, i.e., the importance of considering the full structure of the error in order to reliably accelerate exposure calculations with $L^p$ approximations. Even though PFE is not covered by Proposition \ref{prop:lp_bound}, the result fits intuitively as a limit case where all probability mass $m$ is concentrated in a single point, namely, the considered confidence level.

\begin{figure}\begin{minipage}[t]{0.45\textwidth}
\scalebox{0.85}{
\input{valuefs_NEW.tex}
}
\caption{Value functions $V_t(z)$ (solid) and $U_t(z)$ (dashed) of the two digital call options with slightly different log-strikes and same otherwise.}
    \label{fig:valuefs_digi}
\end{minipage}
\hfill
\begin{minipage}[t]{0.45\textwidth}
\scalebox{0.85}{
\input{testq_NEW.tex}
}
\caption{Quantile functions $Q_{X_t}(p)$ (solid) and $Q_{Y_t}(p)$ (dashed) of exposures ${X_t = V_t(Z_t)}$ and $Y_t = U_t(Z_t)$ for $p \geq 0.99$. $\mathrm{PFE}_{0.99}(X_t)$ is denoted by $\times$, and $\mathrm{PFE}_{0.99}(Y_t)$ by $\bullet$, on the $y$-axis. The gray area is  $0.01 \cdot \mathrm{CES}_{0.99}(X_t)$, and the dashed area is $0.01 \cdot \mathrm{CES}_{0.99}(Y_t)$.} \label{fig:temp2}
\end{minipage}
\end{figure}

\subsection{Finite-sample error bounds}

In practice, one typically cannot directly calculate the exact exposure measurement of interest, $ \rho(X_t) =  \rho(F_{X_t})$, as the distribution $F_{X_t}$ is not available explicitly. Instead, one considers the empirical estimator $\widehat{\rho}(\mathbf{x}_t) = \rho(F_{\mathbf{x}_t})$. In generic accelerations of exposure measurements, one further approximates $\widehat{\rho}(\mathbf{x}_t)$ by $\widehat{\rho}(\mathbf{y}_t)$, where $\mathbf{y}_t$ is the data sample obtained when $V_t$ is replaced by its approximation $U_t$ in full re-evaluation.
In this setting, we obtain the following finite sample error bounds, which are shown in Appendix~\ref{app:proof_emp_lp}.

\begin{proposition}\label{prop:emp_lp}
Let $\mathbf{X}_t = (X_t^i)_{i = 1}^n$ and $\mathbf{Y}_t = (Y_t^i)_{i = 1}^n$ be the random samples of the exposures $X_t=V_t(Z_t)^+$ and $Y_t=U_t(Z_t)^+$ , i.e., $X_t^i =  V_t(Z_t^i)^+$ and $Y_t^i = U_t(Z_t^i)^+$ where $Z_t^i \stackrel{i.i.d.}{\sim} Z_t$, and we assume $Z_t$ admits a bounded density $f_{Z_t}$. Let $1 \leq p < \infty$ { and $\eta \in (0,1]$}. With a probability of at least $1-\eta$ we have:
\begin{enumerate}
\item[(a)] For any law-invariant exposure measure $\rho$: 
\[
\left| \widehat \rho (\mathbf{X}_t) - \widehat \rho (\mathbf{Y}_t) \right| \leq \left(\frac{\|f_{Z_t}\|_\infty}{{\eta}}\right)^{1/p} n^{1/p} \|V_t - U_t\|_{L^p(\mathbb{R})}.
\]
\item[(b)] For any quantile-based exposure measure $\rho_m$ with absolutely continuous $m$ with density $f_m$:
\begin{equation*}
\begin{split}
\left| \widehat \rho_m (\mathbf{X}_t) - \widehat \rho_m (\mathbf{Y}_t) \right| \leq \|f_m\|_{L^q(0,1)} &\|f_{Z_t}\|_\infty^{1/p}  \|V_t - U_t\|_{L^p(\mathbb{R})}\\
&+ \sqrt[2p]{\frac{-\ln(\eta)}{n}} \|f_m\|_{L^q(0,1)} \| V_t - U_t\|_\infty,
\end{split}
\end{equation*}
where $1/p+1/q=1$.
\end{enumerate}
\end{proposition}

While the estimate (a) holds in general, the factor $n^{\frac1p}$ reduces
 the quality of the estimate. On the other hand, estimate (b) formally requires $V_t-U_t$ to be bounded, but the influence of the uniform norm reduces with larger number of samples. Thus even in cases where the uniform error is large, the inequality can provide a reasonable estimate based on the $L^p$ norm.

\section{Numerical experiments}\label{sec:num_exp}

The theoretical results summarised in Corollary \ref{cor:cheb_conv} show the efficiency of the Chebyshev interpolation for accelerating simulation-based exposure calculations. In this section we present numerical experiments confirming these results and highlighting further practical applications of the proposed method: 
First, we replace a variety of numerical pricers in the full re-evaluation with their Chebyshev approximations, and discuss the obtained run-time gains and accuracies for EE, PFE and CES. Second, we propose, implement and discuss the adaptive choice of the interpolation degree in order to maximize run-time gains for any given number of simulations. Third, in order to highlight how the method can be tailored to more advanced counterparty risk settings we employ fast and accurate Chebyshev approximations of Greeks to accelerate CVA sensitivity and dynamic initial margin computations.

\subsection{Numerical results of accelerated exposure calculations}\label{subsec:num_res}

In the first set of numerical experiments, we assess the efficiency of the proposed method by evaluating accuracy and run-time gains over a wide range of practically relevant settings. We test the performance for $12$ option pricing problems covering model features such as jumps and stochastic volatility, and option features such as non-linearity, payoff discontinuity and path-dependency. Specifically, we select European, digital, barrier and American options, in the BSM, MJD and HSV models, calibrated to market data.

In the following, we present the key elements of our experimental setup.  The payoff and model parameters used in the experiment are given in Appendix \ref{app:params}.  Calculation and implementation specifics are deferred to Appendix \ref{sec:calc_met_int_dom_deg}.

To analyse the accuracy of Chebyshev exposure profiles we compare them to state-of-the-art full re-evaluation benchmarks. For pricing of European, barrier, and American options in the BSM and HSV models we employ finite-difference PDE solvers, as implemented in QuantLib \cite{quantlib}. Specifically, the Douglas scheme is used for European, barrier, and American options in the BSM model.  In the HSV model, for European and American options the modified Craig-Sneyd scheme is used, and the Hundsdorfer-Verwer scheme for the barrier option, see \cite{hout2008} for a discussion of the different schemes. In the MJD model, and for digital options in the BSM and HSV models, we use the COS method \cite{fang2009euro,fang2009} which goes beyond industry standard benchmarks.

We first draw realizations of the risk factors $\{(z_{t_u}^i)_{u = 1}^m\}_{i=1}^n$ for each of the three considered stochastic models. For this, we use $n = 10^4$ sample paths and $ m = 52$ time steps, employing the Euler discretization of the underlying stochastic differential equations. To assess the statistical accuracy, we compute the relative length of the corresponding $95\%$ confidence interval for each considered empirical estimator, see Appendix \ref{sec:calc_met_int_dom_deg}. The confidence intervals are derived from the assumption of a normally distributed exposure estimator with mean equal to the estimator's value and variance equal to its empirical variance, see Appendix \ref{sec:exp_mea_exa}.

Next, in each of the 12 option pricing setups we calculate path-wise exposures (taking early-exercise features into account as described in Appendix \ref{sec:calc_met_int_dom_deg})
\begin{equation}\label{eq:pwe}
\left(\mathbf{x}_{t_u}=\left(V_{t_u}(z^i_{t_u})^+\right)_{i=1}^n\right)_{u=1}^m \hspace{0.25cm}\text{ and }\hspace{0.25cm} \left(\mathbf{y}_{t_u}=\left(U_{t_u}(z^i_{t_u})^+\right)_{i=1}^n\right)_{u=1}^m 
\end{equation}
with the reference pricer $V_{t_u}$ and its Chebyshev approximation $U_{t_u} \approx V_{t_u}$. To evaluate the run-time improvements, we measure the execution times for both calculations outlined in equation \eqref{eq:pwe} and determine the speed-up factors by dividing the time taken for the calculation using $V_{t_u}$ by the time using $U_{t_u}$. Note that the time to evaluate $\widehat{\rho}$ from the sample prices is negligible in comparison to the time of repeated evaluations of the pricer or its approximation. Therefore, the speed-up factors are evaluated independently of the considered exposure measure $\rho$.

We quantify the \emph{acceleration error}, that is, the error induced by the Chebyshev acceleration, as the maximum relative error over the empirical time-profile,
\begin{equation}\label{eq:max_rel_errs}
\varepsilon_{\widehat{\rho}} = \max_{u=1,...,m} \left| \widehat{\rho}(\mathbf{x}_{t_u}) - \widehat{\rho}(\mathbf{y}_{t_u})\right| / \widehat{\rho}(\mathbf{x}_{t_u}),
\end{equation}
for three different exposure measures $\rho$: Expected Exposure $\mathrm{EE}$, Potential Future Exposure $\mathrm{PFE}_{0.95}$, and Credit Expected Shortfall $\mathrm{CES}_{0.95}$.

\begin{figure}
\centering  
  	% This file was created with tikzplotlib v0.10.1.
\begin{tikzpicture}

\definecolor{darkgray176}{RGB}{176,176,176}
\definecolor{lightblue161201244}{RGB}{161,201,244}
\definecolor{lightgray204}{RGB}{204,204,204}
\definecolor{mediumturquoise100181205}{RGB}{100,181,205}
\definecolor{steelblue76114176}{RGB}{76,114,176}

\begin{axis}[
legend entries={BSM,
                MJD,
                HSV},
legend cell align={left},
legend style={
  fill opacity=0.8,
  draw opacity=1,
  text opacity=1,
  at={(0.05,0.9)},
  anchor=north west,
  draw=lightgray204
},
tick align=outside,
x grid style={darkgray176},
xmin=-0.3125, xmax=3.8125,
xtick style={color=white},
xtick={0.25,1.25,2.25,3.25},
xticklabels={European,Digital,Barrier,American},
ytick pos=right,
y grid style={darkgray176},
ymin=0, ymax=241.5,
ytick style={color=black},
ytick={14,50,87,100,150,200,230},
yticklabels={14,50,\textcolor{gray}{87},100,150,200,\textbf{230-fold speed-up}},
width=10cm,
height=6cm
]

\addlegendimage{only marks, mark=square*,draw=white,fill=lightblue161201244}
\addlegendimage{only marks, mark=square*,draw=white,fill=mediumturquoise100181205}
\addlegendimage{only marks, mark=square*,draw=white,fill=steelblue76114176}

\draw[dashed] (axis cs: -0.3125, 230) -- (axis cs: 3.8125, 230);
\draw[dashed] (axis cs: -0.3125, 14) -- (axis cs: 3.8125, 14);
\draw[dashed,darkgray176] (axis cs: -0.3125, 87) -- (axis cs: 3.8125, 87);

\draw[draw=lightgray204,fill=lightblue161201244] (axis cs:-0.125,1) rectangle (axis cs:0.125,83);
%\addlegendimage{ybar,ybar legend,draw=white,fill=lightblue161201244}
%\addlegendentry{BSM}

\draw[draw=lightgray204,fill=lightblue161201244] (axis cs:0.875,1) rectangle (axis cs:1.125,69);
\draw[draw=lightgray204,fill=lightblue161201244] (axis cs:1.875,1) rectangle (axis cs:2.125,88);
\draw[draw=lightgray204,fill=lightblue161201244] (axis cs:2.875,1) rectangle (axis cs:3.125,80);
\draw[draw=lightgray204,fill=mediumturquoise100181205] (axis cs:0.125,1) rectangle (axis cs:0.375,97);
%\addlegendimage{ybar,ybar legend,draw=white,fill=mediumturquoise100181205}
%\addlegendentry{MJD}

\draw[draw=lightgray204,fill=mediumturquoise100181205] (axis cs:1.125,1) rectangle (axis cs:1.375,93);
\draw[draw=lightgray204,fill=mediumturquoise100181205] (axis cs:2.125,1) rectangle (axis cs:2.375,174);
\draw[draw=lightgray204,fill=mediumturquoise100181205] (axis cs:3.125,1) rectangle (axis cs:3.375,230);
\draw[draw=lightgray204,fill=steelblue76114176] (axis cs:0.375,1) rectangle (axis cs:0.625,59);
%\addlegendimage{ybar,ybar legend,draw=white,fill=steelblue76114176}
%\addlegendentry{HSV}

\draw[draw=lightgray204,fill=steelblue76114176] (axis cs:1.375,1) rectangle (axis cs:1.625,14);
\draw[draw=lightgray204,fill=steelblue76114176] (axis cs:2.375,1) rectangle (axis cs:2.625,36);
\draw[draw=lightgray204,fill=steelblue76114176] (axis cs:3.375,1) rectangle (axis cs:3.625,46);

\end{axis}

\end{tikzpicture}%
  	\caption{Run-time gains of the Chebyshev acceleration. The figure displays \textbf{\emph{factors}} by which the approach is faster than the standard full re-evaluation in calculation of path-wise exposures. The mean speed-up factor of 87 implies that computations where the full re-evaluation takes hours, take minutes with the accelerated approach.}\label{fig:SUF} 
\end{figure}
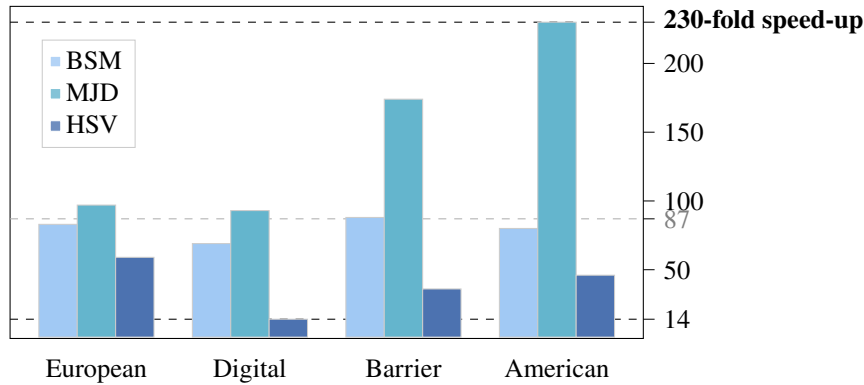

\begin{table}[]
\centering
\begin{adjustbox}{max width=\textwidth}
\begin{tabular}{l@{\extracolsep{4pt}}l@{\extracolsep{4pt}}cccccc}
\toprule
      &          & \multicolumn{2}{c}{\textbf{EE}}                 & \multicolumn{2}{c}{\textbf{PFE}}         & \multicolumn{2}{c}{\textbf{CES}}            \\ \cmidrule{3-4} \cmidrule{5-6} \cmidrule{7-8}
\textbf{Model} & \textbf{Option}   & $\varepsilon$ & MC             & $\varepsilon$ & MC             & $\varepsilon$ & MC            \\ \midrule
BSM   & European\hspace{0.2cm} & $7.5 \times 10^{-4}$           & $4.6 \times 10^{-2}$ & $1.5 \times 10^{-3}$           & $5.5 \times 10^{-2}$ & $9.1 \times 10^{-4}$           & $5.1 \times 10^{-2}$ \\
MJD   & European & $4.2 \times 10^{-4}$           & $9.3 \times 10^{-2}$ & $2.0 \times 10^{-3}$           & $6.5 \times 10^{-2}$ & $3.2 \times 10^{-4}$           & $7.6 \times 10^{-2}$ \\
HSV   & European\vspace{0.2cm} & $1.3 \times 10^{-4}$           & $3.7 \times 10^{-2}$ & $8.3 \times 10^{-4}$           & $2.8 \times 10^{-2}$ & $2.8 \times 10^{-4}$           & $4.4 \times 10^{-2}$ \\
BSM   & Digital  & $5.9 \times 10^{-5}$           & $4.6 \times 10^{-2}$ & $8.4 \times 10^{-4}$           & $3.6 \times 10^{-3}$ & $1.9 \times 10^{-4}$           & $3.2 \times 10^{-4}$ \\
MJD   & Digital  & $4.3 \times 10^{-4}$           & $1.6 \times 10^{-2}$ & $6.9 \times 10^{-4}$           & $2.1 \times 10^{-3}$ & $2.3 \times 10^{-4}$           & $2.5 \times 10^{-4}$ \\
HSV   & Digital\vspace{0.2cm}  & $6.8 \times 10^{-4}$           & $3.1 \times 10^{-2}$ & $4.2 \times 10^{-3}$           & $1.8 \times 10^{-2}$ & $2.0 \times 10^{-3}$           & $2.7 \times 10^{-3}$ \\
BSM   & Barrier  & $1.3 \times 10^{-4}$           & $4.4 \times 10^{-2}$ & $1.5 \times 10^{-3}$           & $6.9 \times 10^{-3}$ & $3.9 \times 10^{-3}$           & $1.0 \times 10^{-2}$ \\
MJD   & Barrier  & $2.8 \times 10^{-4}$           & $8.7 \times 10^{-2}$ & $1.9 \times 10^{-2}$           & $6.0 \times 10^{-2}$ & $5.9 \times 10^{-4}$           & $5.6 \times 10^{-2}$ \\
HSV   & Barrier\vspace{0.2cm}  & $5.2 \times 10^{-4}$           & $3.6 \times 10^{-2}$ & $1.6 \times 10^{-2}$           & $2.0 \times 10^{-2}$ & $1.1 \times 10^{-2}$           & $1.2 \times 10^{-2}$ \\
BSM   & American & $6.0 \times 10^{-6}$           & $9.4 \times 10^{-3}$ & $1.2 \times 10^{-5}$           & $3.5 \times 10^{-2}$ & $5.0 \times 10^{-6}$           & $4.1 \times 10^{-2}$ \\
MJD   & American & $4.7 \times 10^{-3}$           & $1.4 \times 10^{-1}$ & $2.0 \times 10^{-5}$           & $4.1 \times 10^{-2}$ & $4.0 \times 10^{-6}$           & $3.8 \times 10^{-2}$ \\
HSV   & American & $8.9 \times 10^{-3}$           & $1.4 \times 10^{-1}$ & $1.1 \times 10^{-2}$           & $9.9 \times 10^{-2}$ & $1.2 \times 10^{-2}$           & $7.7 \times 10^{-2}$ \\ \bottomrule
\end{tabular}
\end{adjustbox}
\caption{Acceleration of exposure calculations with the Chebyshev interpolation. In addition to the maximum relative errors $\varepsilon$ in exposures time-profiles \eqref{eq:max_rel_errs}, the table displays the corresponding lengths of $95\%$ Monte Carlo confidence intervals relative to the values of the estimates (MC). \vspace{-0.25cm}}
\label{tab:my-table}
\end{table}

Approximations $U_{t_u} \approx V_{t_u}$ are obtained by the Chebyshev interpolation with the domain splitting for each time point except for the last one where the value $V_T$ is equal to the analytically known payoff. In this section, we choose to demonstrate the run-time gains at the level of accuracy of the simulation-based methodology for a pre-specified number of scenarios, i.e., we aim for an acceleration error to be below the MC error (both measured in the relative sense) in all cases.

To this end, we begin with the interpolation degree 8 on each subdomain (for 2D cases, i.e., in the HSV model, this means 8 in each dimension), and hand-tune it if necessary to reach the MC accuracy. While this is done for comparability reasons, in the next section we propose a systematic approach to adaptively choose the interpolation degree. Overall, the interpolation with degree 8 requires $2 \cdot 9 \cdot 51 = 918$ evaluations of a reference pricer in 1D cases (BSM, MJD), and  $ 2 \cdot 81 \cdot 51 = 8\,262$ evaluations in 2D cases (HSV). We increase the interpolation degrees in three cases,  resulting in $936$ pricer evaluations for the barrier option in the BSM model, and $16\,950$ and $13\,254$ evaluations respectively for digital and barrier options in the HSV model. The standard full re-evaluation benchmark requires $51 \cdot 10\,000 = 510\,000$ reference pricer evaluations.

Table \ref{tab:my-table} presents the maximum relative acceleration errors and the MC accuracies for the times at which the former are realized. We note that in all cases the acceleration errors are below the lengths of the MC confidence intervals, as the number of Chebyshev nodes was chosen accordingly. In some cases, the acceleration errors are even smaller by several orders of magnitude. For example, for the American option in the BSM model they are smaller by a factor of $10^{-3}$. Overall, the relative acceleration errors are below 1\%, and typically significantly lower, for all but 5 of the 36 considered exposure measurements, where they are still less than 2\%. Note that this holds true across all the time points, as we report the maximums over the time-profiles.

The speed-up factors observed in the experiments are reported in Figure \ref{fig:SUF}. 
 The observed run-times and thus speed-up factors account for the calls to pricers as well as other required steps, such as modifying the path-wise exposures to accommodate for the early-exercise feature in the case of American options. 

For each option type, we observe the highest run-time gains for the MJD and lowest for the HSV model. We attribute these differences to differences in the calculation times of pricers and their interpolations. Clearly, reference pricers for the MJD model are more complex than the ones for the BSM model resulting in higher run-time gains of the interpolation approach. While the reference pricers also are more complex for the HSV than for the BSM model, we here observe an additional effect. As the underlying pricing problem in the HSV model is two-dimensional, more calls to the reference pricers are needed to construct the approximation, which reduces the run-time gains.  Nevertheless, significant speed-ups of factor 14 to factor 59 are achieved in all cases.
 
For the MJD model we observe speed-up factors of around 90 for the European and digital options, 174 for the barrier and 230 for the American option. We attribute the difference in speed-up factors to the varying complexity of the reference solver. The higher the run-time for the reference solver, the higher the speed-up by the interpolation. For the barrier and American options, the reference solvers utilise a loop to account for path-dependency, increasing their run-times. An additional iterative root-finding algorithm is required to determine the early-exercise boundary in the American case. In comparison, this is not necessary for the evaluation of the Chebyshev interpolation and the run-time gain is highest in this case. 
 
In the BSM model we observe comparable speed-up factors of 69-88 across all options as the employed reference pricers are comparably efficient.

It is also interesting to note that the speed-up factors for the HSV model are increasing from the digital over the barrier to the American option, and finally it is highest for the European option. For the digital and barrier options in the HSV model, the number of Chebyshev nodes had to be increased to reach the target accuracy reducing the run-time gains. For the digital option this is due to the comparably high Monte-Carlo accuracy, especially in the case of CES, while for the barrier option we account this to the high curvature of the value function. The COS method is used for the digital option, which is of highest efficiency, explaining why the corresponding speed-up factor is smallest. Comparing the European and American case in the HSV model, we observe a smaller speed-up in the American case due to additional time required to accommodate for the early-exercise feature in the path-wise exposures in the latter case.

To summarize the results, we have observed a range of considerable speed-ups throughout the choice of different option types and models, from a factor of 14 to extraordinary speed-ups of a factor of 230, while maintaining the accuracy of the simulation based methodology. The transparency of the method facilitates the understanding of differences in the speed-ups in terms of the different model, option and procedural features adding to the attractiveness of the method. We conclude that the Chebyshev acceleration method yields fast and accurate exposure calculations even in complex cases. It thus helps to avoid oversimplification of risk to achieve practically feasible run-times.

\subsection{Preliminary investigation into adaptive choice of interpolation degrees}\label{subsec:ad_ch}

Propositions \ref{prop:uniform_bound_11} and \ref{prop:uniform_bound_12} provide decompositions of the overall error into simulation noise and function approximation error.
They enable an adaptive choice of the number of Chebyshev nodes according to a target accuracy. Here, we choose the target accuracy based on the MC noise for a given number of simulations.
We present a preliminary investigation of 
the adaptive choice, using 
an a posteriori error estimate of the function approximation.

Given a fixed number of simulations, we want to adaptively choose the number of Chebyshev points efficiently based on a posteriori interpolation error estimates. To do so, we implement the following procedure:

\begin{enumerate}
	\item Estimate the widths of the MC confidence intervals for the exposure measures EE, PFE, CES, and set the target accuracy $\varepsilon$ for the interpolation as the smallest one.
	\item Iteratively increase the interpolation degree. Starting with the approximation of degree 1:
	\begin{enumerate}
		\item[2.1] Double the interpolation degree.
		\item[2.2] Estimate the function approximation error $\varepsilon'$ as the maximum difference between the current and preceding approximation on 100 random points from the interpolation domain.
		\item[2.3] If $\varepsilon'<\varepsilon$, select the preceding approximation and stop.
	\end{enumerate}
	\item Evaluate the selected approximation path-wise and extract the empirical exposure measures.
\end{enumerate}

In the first step, we estimate the confidence intervals from path-wise exposures obtained using a piecewise linear approximation of the pricer, as in the previous section.
In the iterative step we choose to double the interpolation degree in order to exploit the resulting nested structure of the Chebyshev nodes, i.e., 
\[\{\cos(\pi k / N)\}_{k=0,...,N} \subset \{\cos(\pi k /(2N))\}_{k=0,...,2N}.\]
This reduces the number of calls to expensive pricers. For the same reason, in step 2.2 we substitute the reference pricer in the error estimation by the higher-level approximation.

Numerical results for the case of a European call option in the BSM model highlight a favorable performance of the adaptive choice. The errors induced by the Chebyshev interpolation are below the MC errors estimated from the full re-evaluation for all times of each MC simulation. The speed-up factors for each case are reported in Table \ref{tab:suf-ad}.

\begin{table}[H]
\centering
\begin{tabular}{lrrrr}
\toprule
\# paths        & $1\,250$ & $2\,500$ & $5\,000$ & $10\,000$ \\ \midrule
speed-up factor & 25       & 34       & 44       & 56       \\ \bottomrule
\end{tabular}
\caption{Speed-up factors for different numbers of MC paths using the hierarchical adaptive algorithm to determine the interpolation degree.}
\label{tab:suf-ad}
\end{table}

We observe good speed-up factors for all cases in Table \ref{tab:suf-ad}, with an increased run-time gain for larger numbers of simulation paths, confirming that our specification of the a posteriori adaptive algorithm is appropriate.
The generic structure of the proposed hierarchical algorithm makes us confident that its good performance extends to other cases of exposure calculations. Therefore, we propose it as a suitable tool when aiming for highest practically possible run-time gains for a given number of MC paths.

\subsection{Sample applications of accelerated Greeks to counterparty credit risk}\label{sec:cheby_greeks}

In order to illustrate the use of the method beyond the framework considered so far, we present preliminary investigations into two counterparty risk tasks where sensitivities, i.e., the Greeks, play a key role.
In addition to the good approximation of the price, the Chebyshev interpolation directly yields good approximations of the Greeks. Namely, 
the derivative of the approximation obtained by the Chebyshev interpolation is given by the explicit formula in terms of the interpolation coefficients.
In this section we numerically illustrate how the available approximations of Greeks can be exploited to enhance advanced counterparty credit risk calculations. The examples are kept illustrative, with selected cases in which the required derivatives of the prices are available analytically and thus serve as reliable benchmarks.

\subsubsection{Estimating CVA's delta}\label{sec:cva_sens}

Credit valuation adjustment is an adjustment made to the theoretical risk-neutral value of the trade to account for a counterparty risk \cite{green2015,gregory2020xva}. In the following we consider the case of an option whose price process $V_t(S_t)$ is independent of the default event, and assume recovery rates to be equal to zero for simplicity. Then, the unilateral credit valuation adjustment (CVA) for a trade $V_t$ depending on a risk factor $S_t = s_0 e^{\mathcal{X}_t}$ with the time-horizon $T$ is given by
	\begin{equation*}
	\text{CVA}(s_0) = \int_0^T e^{-rt} \mathbb{E}\left[V_t(s_0 e^{\mathcal{X}_t})  \right] d\mathrm{PD}(t),
	\end{equation*}
where $\mathrm{PD}(t)$ is the risk-neutral probability of counterparty default before time $t$. The above integral is typically discretized over a time-grid $\{t_u\}_{u=0}^m$ in practice,
\begin{equation*}
	\text{CVA}(s_0) \approx \sum_{u=1}^m e^{-rt_u} \mathbb{E}\left[V_{t_u}(s_0 e^{\mathcal{X}_{t_u}}) \right] \left( \mathrm{PD}(t_u) - \mathrm{PD}(t_{u-1})\right).
\end{equation*}

CVA is often hedged, which requires its derivative with respect to an underlying risk factor. Under appropriate conditions, ensuring that the order of differentiation and expectation can be interchanged, we have
\begin{equation*}
\begin{aligned}
\mathrm{CVA}'(s_0) &\approx \sum_{u=1}^m e^{-rt_u} \mathbb{E}\left[\frac{d }{d s_0} V_{t_u}(s_0 e^{\mathcal{X}_{t_u}})\right] \left( \mathrm{PD}(t_u) - \mathrm{PD}(t_{u-1})\right)\\
&= \sum_{u=1}^m e^{-rt_u}\mathbb{E}\left[  V_{t_u}'(s_0 e^{\mathcal{X}_{t_u}}) e^{\mathcal{X}_{t_u}} \right] \left( \mathrm{PD}(t_u) - \mathrm{PD}(t_{u-1})\right).
\end{aligned}
\end{equation*}

If $S_t/ S_0$ is a stochastic process that does not depend on $S_0$, the expectation in the above equation can be approximated in the simulation-based framework as
\begin{equation*}
\mathbb{E}\left[ V_{t_u}'(s_0 e^{\mathcal{X}_{t_u}}) e^{\mathcal{X}_{t_u}} \right]\approx \frac{1}{n} \sum_{i=1}^n V_{t_u}'\left(s_{t_u}^i\right) \cdot \frac{s_{t_u}^i}{s_0}.
\end{equation*}
For details on path-wise derivative estimations see \cite[Chapter~7.2]{glasserman2004monte}.

To summarize, the CVA's delta can be estimated within the simulation-based framework as
\begin{equation}\label{eq:sb_cva_sens}
\mathrm{CVA}'_{\mathrm{MC}}(s_0) = \frac{1}{n} \sum_{u=1}^m  \left[ e^{-rt_u} \left( \mathrm{PD}(t_u) - \mathrm{PD}(t_{u-1})\right)  \sum_{i=1}^n V_{t_u}'\left(s_{t_u}^i\right) \cdot \frac{s_{t_u}^i}{s_0} \right] .
\end{equation}

We illustrate the numerical efficiency of the Chebyshev derivatives for the calculation of \eqref{eq:sb_cva_sens} with an example.

\begin{example}
We consider the case of a European call option in the BSM model with the same parameters as in the previous experiment. We assume that $\mathrm{PD}$ is the uniform distribution over $[0, T]$ for simplicity.

First, we calculate the benchmark value of \eqref{eq:sb_cva_sens} with the analytically available $V_{t_u}'$
\[\mathrm{CVA}'_{\mathrm{MC}}(s_0) = 0.5600,\]
which took 34.7 seconds.

Next, we calculate the approximation of $\mathrm{CVA}'_{\mathrm{MC}}(s_0)$ using the Chebyshev sensitivities $U_{t_u}'$ in place of $V_{t_u}'$ in \eqref{eq:sb_cva_sens}, where $U_{t_u} = I_N(V_{t_u})$ is the approximation of $V_{t_u}$ with the Chebyshev interpolation of degree $N=16$, and obtain
\[\mathrm{CVA}'^{\mathrm{Cheby}}_{\mathrm{MC}}(s_0) = 0.5599,\]
with a run-time of only 32 seconds.

\end{example}

We observe that employing the derivatives of the Chebyshev interpolation in the simulation-based CVA sensitivity calculation introduces a negligible error only, and even slightly reduces the run-time compared to using the analytical formula.

Based on the numerical results and their discussion in Section \ref{subsec:num_res}, we can expect the run-time gains to be higher for higher run-times to compute the Greeks as main effect. Moreover, we expect a similar contribution of the different complexities of models, options and procedures to the actual run-time reduction as discussed in Section \ref{subsec:num_res}. Therefore we can expect considerable speed-ups for more realistic models and a wide range of situations where Greeks need to be calculated numerically. To confirm the applicability to various situations we consider an example of another counterparty risk task in the next section.

\subsubsection{Sensitivity-based initial margins}

Recent regulatory developments indicate a shift to mandatory collateralization of non-centrally cleared derivative trades \cite{bcbs2020}.
To reduce the counterparty exposure, the variation margin $\mathrm{VM}_{t+\delta} = V_{t+\delta}(Z_{t+\delta}) - V_{t}(Z_{t}) = \Delta_{t+\delta} $ is to be regularly posted between the counterparties. 
An additional collateral, the initial margin, is required to account for potential losses over the time period $[t, t+\delta)$. The forward initial margin is typically defined as $\mathrm{IM}_t = Q_{\Delta_{t+\delta}|\mathcal{F}_t}(\alpha)$, for some confidence level $\alpha$ (e.g., $\alpha = 0.99$), see, e.g., \cite{Andersen2017IM,Andersen2017MPOR} for details on IM and VM.

Due to funding costs, posting initial margins requires adjusting the value of the trade. This is known as the (initial) margin valuation adjustment (MVA),
\begin{equation}\label{eq:MVA}
\mathrm{MVA} =  \int_0^T \mathbb{E}\left[ \mathrm{IM}_t \right] \mathrm{FS}(t)dt,
\end{equation} 
where $\mathrm{FS}(t)$ is the funding spread, see, e.g., \cite[Chapter~20]{gregory2020xva}.

In the fully simulation-based framework, computing the MVA is based on a nested simulation at each considered time $t$. The outer simulation of $z_t^i$, $i = 1,...,n$, is required to calculate the expectation, and the inner simulation is used to evaluate $\mathrm{IM}_t^i = Q_{\Delta_{t+\delta}|Z_t = z_t^i}(\alpha)$ for each of the outer paths.
A general approach to approximate $\mathrm{IM}_t^i$ is: 1) the risk factors are simulated from the distribution $Z_{t+\delta}|Z_t = z_t^i$ resulting in the samples $\{ z_{t+\delta}^{ij} \}_{j=1}^p$ for each $i$; 2) the changes in the derivative value $\Delta_{t+\delta}^{ij} = V_{t+\delta}(z_{t+\delta}^{ij}) - V_{t}(z_{t}^i)$ are then calculated; and 3) the quantile $\mathrm{IM}_t^i$ is extracted from the empirical distribution induced by the sample $\boldsymbol\Delta_{t+\delta}^i = \{ \Delta_{t+\delta}^{ij} \}_{j=1}^p$.

This nested Monte-Carlo approach requires $mnp$ calls to a derivative pricing function, where $m$ is the number of exposure calculation time points, $n$ the number of paths used for exposure calculation, and $p$ the number of paths for the calculation of each initial margin. To overcome the huge computational challenges associated with a nested Monte-Carlo approach, the International Swaps and Derivatives Association (ISDA) has proposed a standard methodology based on sensitivities \cite{isda2021}. Namely, instead of obtaining the empirical samples $\boldsymbol\Delta_{t+\delta}^i$ by calling the value function on inner simulation paths, they propose to approximate $\boldsymbol\Delta_{t+\delta}^i$ by $\tilde{\boldsymbol\Delta}_{t+\delta}^i$ using the derivative's delta:
\begin{equation}\label{eq:isda_pl}
\Delta_{t+\delta}^{ij} \approx \tilde{\Delta}_{t+\delta}^{ij} = V_t'(z_t^i)(z_{t+\delta}^{ij} - z_{t}^i).
\end{equation}

Time-discretizing the MVA \eqref{eq:MVA}, and denoting ${\mathrm{IM}_t^i}_\mathrm{ISDA} = Q_{\tilde{\boldsymbol\Delta}_{t+\delta}^i}(\alpha)$, we have
\begin{equation}\label{eq:isda_mva}
\mathrm{MVA}_{\mathrm{ISDA}}= \sum_{u=1}^m  \left[ \mathrm{FS}(t_u)(t_u - t_{u-1}) \frac{1}{n} \sum_{i=1}^n {\mathrm{IM}_t^i}_\mathrm{ISDA}  \right] 
\end{equation}
as the practical estimate of MVA within the simulation-based ISDA approach employing sensitivities.

In the following example we compare the calculation of the above formula with analytically available delta and Chebyshev approximations.

\begin{example}
We again consider the case of a European call option in the BSM model with the same parameters as before. We will assume that $\mathrm{FS}$ is constant and equal to 0.01 over $[0, T]$ for simplicity.

First, we calculate \eqref{eq:isda_mva} employing the analytically available exact $V_{t_u}'$ in \eqref{eq:isda_pl} and obtain 
\[\mathrm{MVA}_{\mathrm{ISDA}} = 2.0114,\]
with a run-time of 73.3 seconds.

Next, we calculate the approximation of $\mathrm{MVA}_{\mathrm{ISDA}}$ using the derivatives of the Chebyshev interpolation $U_t'$ in place of $V_t'$ in \eqref{eq:isda_pl}, where $U_t = I_N(V_t)$ is the approximation of $V_t$ with the Chebyshev interpolation of degree $N=16$, and obtain
\[\mathrm{MVA}_{\mathrm{ISDA}}^{\mathrm{Cheby}} = 2.0168,\]
with a run-time of 72 seconds.
\end{example}

We again observe that employing the derivatives of the Chebyshev interpolation introduces only negligible errors, and slightly reduces the run-time compared to using the analytical formula.

To summarize, our numerical experiments cover three practical applications:  evaluations of exposure measures, CVA's delta and initial margins. The results consistently show that the method is ready to serve to considerably speed-up  a wide range of basic and advanced counterparty risk tasks.

\section{Related literature}\label{sec:literature}

Counterparty credit risk inherently comprises two risk components, credit risk and market risk, referring to the possibility of the counterparty's default and the uncertainty of the consequent loss, respectively. In this work we have focused on market risk side, while for related works on credit risk aspect, particularly interesting functional approximations for large portfolios, we refer to~\cite{bo2014,giesecke:13,giesecke:15}.

Quantifying market risk is required for both pricing and risk management, for a comprehensive discussion, see \cite{green2015,gregory2020xva}. Besides integrating counterparty risk into pricing models, rendering them non-linear and high-dimensional, see, e.g., \cite{crepey2016,crepey2015a,crepey2015b} on modelling, and \cite{crepey2018,andersson:25,reisinger2023} on computational aspects, there is the exposure-based approach. In the latter, the counterparty-risk-free value of a derivative is calculated classically, and then adjusted based on suitable exposure measurements. The industry standard for exposure measurements is based on Monte Carlo simulations for generating future market scenarios.
Since our proposed method follows this general approach we mainly focus on the literature discussing simulation-based counterparty exposure calculations.

In practice, MC for scenario generation in the counterparty risk setting needs to be combined with numerical pricing techniques. For example, pricing via simulation leads to nested MC \cite{broadie2011,Gordy2010} and yields significant computational challenges which has given rise to several approaches to accelerate the calculations. 
{The high computational costs of full re-evaluation are twofold. First,  evaluations of the pricing functions for each simulated scenario
require considerable numerical efforts. Second, the large number of scenarios is usually required
to obtain a sufficiently precise estimate.
Correspondingly, there are two different directions
for an attempt to accelerate the calculation
while maintaining the accuracy:
1) speed up evaluations of pricing functions; or
2) decrease the number of required scenarios (via the variance reduction or quasi-Monte Carlo techniques,
see, e.g., Chapters 4 and 5 in \cite{glasserman2004monte}, \cite{Glasserman2000VarianceRT,renzitti2020}).}
Regression-based approaches suggest themselves for the acceleration due to their implementational efficiency and flexibility.

For example, \cite{Schftner2008OnTE} show how to extend the least-squares Monte Carlo (LSMC) approach of Longstaff and Schwartz \cite{Longstaff} to calculate exposures and numerically illustrate the performance of the resulting algorithm. In \cite{karlsson2016} and \cite{feng2016}, the stochastic grid bundling method (SGBM) of \cite{jain2015} is applied for exposure calculations and compared to the LSMC algorithm. Their numerical comparison reveals that the noise in the LSMC approach can lead to inaccurate exposures, particularly for tail exposure measures. This called for a thorough theoretical investigation to gain a better understanding of the error in exposure accelerations. The explanation currently available in the literature is that this noise is due to the early-exercise boundary not being approximated well enough. The analysis unfolded in this paper contributes the insight that deficiencies are of a structurally deeper nature. Particularly, we show that unanticipated large exposure errors can even happen in simpler cases without early-exercise features when a good least-squares fit is used.
To clarify these phenomena, we have revealed the dependence of the exposure error on the distribution of the underlyings and the structure of the exposure measure in addition to the $L^2$ function approximation error. This contributes to the understanding of $L^2$ based exposure accelerations, covering both traditional and recent developments, such as SGBM and the deep learning PDE solver of \cite{wunderlich2023}.

Exploiting the theoretical results we have proposed an acceleration method based on the Chebyshev interpolation. Our contribution expands the literature by presenting an acceleration method underpinned by error bounds allowing to control the exposure error through the interpolation error. In contrast to methods \cite{glau19dc,karlsson2016,Schftner2008OnTE,Shen2013ABA} tailoring specific early-exercise pricers for exposure calculations, our method directly leverages the pricing functions already available within the existing pricing and risk infrastructure yielding high operational efficiency.
An approach that is especially related to this work is the Dynamic Chebyshev method to accelerate exposure calculations~\cite{glau2019fast} as both approaches apply Chebyshev polynomials which provide approximations in the uniform norm. The key difference of the current approach is that the Dynamic Chebyshev method is a stand-alone solver, while the method proposed here accelerates existing pricing methods.

Clearly any contribution to accelerate simulation-based exposure computations also applies to the wider realm of financial risk analysis based on MC scenario generation. 
In particular, this means that our analysis can also be used to deepen the understanding of methods in this broader field. For instance, \cite{broadie2015} provide regression-based nested MC method for financial risk estimation. Our analysis based on the function approximation approach complements their analysis based on classical statistical results. While their results apply to a set of risk measures including the probability of a large loss but not value at risk (i.e., potential future exposure), ours include value at risk but not the probability of a large loss.

Methods targeted to risk assessment typically require relatively low accuracy only.
Counterparty exposure calculations, however, are used for both pricing and risk management. Thus, in this setting, a range of target accuracies is of interest, including very precise exposure estimates for pricing purposes.
Our method proposed has the advantage that it can accelerate while keeping any level of accuracy required in counterparty exposure measurement, involving both risk assessment and pricing tasks. 

\section{Conclusion}\label{sec:conclusion}
\label{sec:conc}

We would like to judge the contributions of the paper against the essential list of practical and theoretical requirements the method ought to fulfil. Since financial risk computations are critical to the economy, the reliability of its procedures is of utmost importance.
The reliability requirement entails both \textit{numerical accuracy} and corresponding \textit{proven error bounds}. Rising regulatory demands require  timely computations of the ever increasing number of risk quantities, and so it is crucial that the method  exhibits \emph{fast run-times}. Large financial institutions have diverse derivatives across different markets on their balance sheets. This implies the need to calculate various risk quantities for different product types, and many different models need to be utilized. This makes the \textit{flexibility in model choice, product type, and exposure measure} another requirement for the method. For the method to have a practical impact it is crucial that the costs of its implementation and integration into the existing well-tested infrastructure of financial institutions are small enough.  We refer to this requirement for the method as the \emph{implementational efficiency}.

Employing function approximation theory, we have systematically investigated a conceptually simple approach to speed-up exposure calculations. The approach consists of replacing the expensive pricing function with an efficient alternative. The simplicity of the method and the regularity of pricing functions are key ingredients for the accurate and practically feasible acceleration of exposure calculations. To rigorously assess the effectiveness, we have formalized and analysed the approach in a generic framework admitting different function approximation techniques.

The subtleties discovered in the developed framework reveal that not all approximation methods are equally suitable for the acceleration. The example in Section \ref{sec:example_new} shows that choosing an unsuitable approximation method may result in surprisingly large errors in exposure measurements, which highlights the need for rigorous analysis. While plenty of related works provide the analysis of the pricing function approximation error, the consequent error in exposure measurement is typically only assessed numerically. A systematic analysis of this error had not been provided in the literature to this day.

In order to close this gap, we have derived error bounds for the exposure measurement induced by replacing the reference pricer in terms of the pricing approximation error. In the introduction we have presented a function approximation framework suitable for the rigorous analysis of accelerating counterparty credit exposure calculations.  Propositions \ref{prop:uniform_bound_11} reveals a structural advantage of the uniform norm as it shows that the acceleration error is directly bounded by the function approximation error. As our main result, in Theorem \ref{thm:probErrBnd} we have shown that we can choose the approximation degree, appropriate truncation of domain, and number of simulations while rigorously controlling the overall error and assessing efficiency. Namely, the acceleration achieves an asymptotic speed-up factor of $n^{1/2-\epsilon}$  for any $\varepsilon>0$ with $n$ the number of simulations. In Proposition \ref{prop:lp_bound} we have proved that the induced error is bounded by the pricing error measured in the $L^p$ norm, for $1 \leq p < \infty$ , with the additional dependency on the distribution used for generating market scenarios and the structure of the exposure measure. In the empirical setting of Proposition \ref{prop:emp_lp} we have shown probabilistic finite sample error bounds.

These theoretical findings and the requirements from industry practice listed in the beginning of the section have guided us to design and implement the approach based on the Chebyshev interpolation. In particular, Corollary \ref{cor:cheb_conv} showcases the proven error bounds yielding \emph{theoretical reliability}.
We note that the method maintains separation of pricing and risk tasks, see Scheme \ref{fig:process_flow}.  This eases the integration with the existing infrastructure as any financial institution should have full re-evaluation readily available as a benchmark, which shows the \emph{implementational efficiency} of the method. 
Exploiting this structural advantage, to estimate the performance gain in terms of run-time and accuracy we have evaluated a variety of exposure quantities in different derivatives pricing setups to demonstrate the method's resulting \emph{flexibility across different models, products and exposure measures}. We have employed three stochastic pricing models distinguished by their sophistication levels each calibrated to market data.

We have compared the proposed method to advanced full re-evaluation benchmarks for the key computation of path-wise exposures. For each benchmark, we have employed a suitable efficient pricer, namely, either a finite-difference PDE solver or the COS method \cite{fang2009euro,fang2009}, adapted to the assumption of only access to black-box evaluations of pricers. The results consistently reveal computational savings by orders of magnitude, even exceeding our expectations.  We observe an average speed-up factor of 87 across the 12 considered cases, while maintaining the
order of error of full re-evaluation, see  Figure \ref{fig:SUF}. 
More precisely, maximum relative errors in all cases are below $2\%$, showcasing \emph{numerical accuracy}, see Table \ref{tab:my-table}. The speed-up factors of about 60 mean that computations that used to take hours now take minutes, and those that took minutes now take seconds. This clearly fulfils the requirement of \emph{fast run-times} completing the list of requirements set out in the second paragraph.

For the example of American options in the Heston model, we observed a reduction of the computational time from 11.5 hours to 15 minutes with relative errors below $1.2\%$. We highlight that the computation time in the accelerated Heston case is on the same scale as the time of the finite-difference calculation for the same early-exercise option within the standard BSM model, which takes 10 minutes. The example indicates that the method may well contribute to avoiding oversimplification of risk computations. This is a key step to overcome a bottleneck between advanced mathematical tools for risk management and industry practice which relies on existing well-tested pricing and risk infrastructure.

Overall, the proposed approach provides the required properties of \textit{numerical accuracy}, \textit{proven error bounds}, \textit{fast run-times}, \textit{implementational efficiency}, and \textit{flexibility in model choice, product type, and exposure measure}.

The discussion of efficiency gains in Section \ref{sec:num_exp} showcases how the transparency of the method enables a clear understanding of the composition of the actually observed run-times based on model, option and procedural features. Together with derived error bounds this facilitates identification of the merits and limits of the method and thus also future research.

\section{Outlook}\label{sec:outlook}

On basis of the results obtained, there are multiple promising directions for further development of the method, both in theory and practice. One concrete line of research is exploiting the fact that the method approximates the Greeks very well, along with the price. We have presented preliminary results for two prominent applications in a counterparty credit risk setting, revealing the high potential of the fast and accurate approximations of Greeks. Building upon the promising results achieved within the BSM model in this paper, we believe that the approach yields substantial efficiency improvements when extended to more advanced models. Therefore, further research to explore the efficiency gain of such applications will be beneficial.

Pointing to another research direction, we note that the presented acceleration approach can be combined with variance reduction methods. As the approach exhibits the separation of pricing and risk tasks, we expect incorporating classical variance reduction strategies to provide substantial gains in efficiency with only small implementational costs.

Another benefit of the method which can be further explored is its compatibility with parallel computing.
The computational workload of the simulation-based exposure quantification is pleasingly parallel, i.e., it is straightforward to accelerate it with parallel computing. As the proposed function approximation approach does not modify the computational structure of the simulation-based framework, it can be combined effectively with parallel computing in all discussed settings.

Finally, we offer our outlook on extending the method to address high-dimensional challenges in the context of counterparty risk. 
The Chebyshev interpolation admits a straightforward tensor-based extension to the multivariate case, see \cite{gass2018}, which inherits the appealing theoretical properties from the univariate setting, see also \cite{glau2019improved,potz2020,sauter2011}. Notably, the convergence is again given in terms of the uniform norm so a multivariate analogue of Corollary \ref{cor:cheb_conv} can be derived. A direct application of the tensorised Chebyshev interpolation would lead to an exponential growth of the number of nodal points as the dimension increases, diminishing expected run-time improvements in high-dimensional settings. Here, sparse grid and low-rank tensor techniques are favourable. These have already been shown to effectively circumvent the dimensionality problem in financial applications of Chebyshev interpolation \cite{glau2020,GRZELAK2022,zeron2021}.
Also machine learning methods have the potential to further improve the efficiency of high-dimensional Chebyshev interpolations, for example when being used to construct the low-rank tensor. 
Also on its own machine learning approximations are promising, and the derived results for $L^2$-approximations form a foundation for high-dimensional approaches minimising the mean-squared loss.
We consider tailoring these methodologies for exposure calculations as a promising direction for future research which would open up several avenues for practical applications.

\section*{Acknowledgements}
This work was funded by the EPSRC grant no. EP/T004738/1.\\
This research utilised Queen Mary's Apocrita HPC facility, supported by QMUL Research-IT doi.org/10.5281/zenodo.438045.

We would like to thank the anonymous reviewers for their helpful remarks.

\appendix\normalsize
\section{Full re-evaluation and acceleration through Chebyshev interpolation}\label{sec:cheb_int}

\subsection{Full re-evaluation}\label{sec:fr}

Full re-evaluation is the benchmark approach for obtaining the approximate distribution of counterparty credit exposure, as \cite[Chapter~15.1.4,~p.~414]{gregory2020xva} notes: ``Whilst add-on and analytical approaches still sometimes exist, Monte Carlo simulation of exposure has been considered state-of-the-art for some time.'' The method estimates the distribution of an exposure $X_t = V_t(Z_t)^+$ by an empirical distribution induced by a sample obtained via Monte Carlo (MC) simulation \cite{glasserman2004monte}, in a procedure summarized as follows.

\begin{enumerate}
	\item Samples $\{z_{t}^i\}_{i=1}^n$ of the future market states $Z_t$ are obtained via Monte Carlo simulation at time points of interest.
	\item A derivative is valued for each market scenario, $v_t^i = V_t(z_t^i )$,
yielding the sample $\mathbf{x}_t ={\{x_{t}^i = \max \{ v_t^i, 0 \}\}_{i=1}^n}$ of the exposure $X_t$.
	\begin{enumerate}
            \item[] {For path dependent options, if the option ceases to exist at time-step $t_{u^{*}}$ on $i$-th path, $x_{t_u}^i$ is set to $0$ for all time-steps $u$, with   $u \geq u^{*}$ for barrier options (knock-out) and $u > u^{*}$ for American options (exercise).}
	\end{enumerate}
	\item The distribution $F_{X_t}$ of the exposure is approximated
by the empirical distribution $F_{\mathbf{x}_t}$ induced by the sample $\mathbf{x}_t$.

\end{enumerate}
The name full re-evaluation stems from the fact that the pricing functions need to be called
for every scenario and every time point of interest. Typically, the exposure is of interest at a multitude of times, that is the exposure's time-profile. We note that full re-evaluation can accommodate practical complexities, such as netting and collateral. Furthermore, the simulation-based framework allows to capture a large number of risk factors and their correlations.

\subsection{Chebyshev acceleration}\label{sec:cheby_acc}

In the following, we present the essentials of the Chebyshev interpolation for exposure accelerations. We refer to \cite{trefethen2013} for a comprehensive review of the Chebyshev interpolation and its appealing {theoretical and} practical properties. Recall, the key appeal in the present setting is that it yields a highly efficient function approximation in the uniform norm that is strikingly fast to construct and to evaluate.

The Chebyshev interpolation  of a function $f$ on the interval $\left[-1, 1\right]$ is a polynomial interpolation of degree $N\geq 1$ using the $N+1$ Chebyshev points $\zeta^k = \cos\left(\pi k/N\right)$, $k=0,...,N$.
 The resulting approximation $I_N(f)$ can be written as a sum of the Chebyshev polynomials $T_j(z) = \cos\left(j \arccos(z)\right)$  weighted by the Chebyshev coefficients, i.e.,
\begin{equation*}
I_N(f)(z) = \sum_{j=0}^N  c_j T_j(z), \hspace{0.15cm}\text{for}\hspace{0.15cm} c_j = \frac{2^{\mathbbm{1}_{\{0<j<N\}}}}{N}\sum_{k=0}^N {}{''} f(\zeta^k) T_j(\zeta^k),
\end{equation*}
where $\sum {}{''}$ indicates that the first and last summand are halved. The interpolation of a function $f$ on a general interval $[a,b]$ is achieved by applying the linear transform 
\begin{equation}\label{eq:domain_transform}
\tau : [-1, 1] \to [a,b], \hspace{0.15cm} \tau(z)=b + \frac{a - b}{2}(1-z),
\end{equation}
i.e., by interpolating the function $f \circ \tau$.

We note that the coefficients $c_j$ are given explicitly as a linear transform of the function values at the known interpolation nodes $\zeta^k$ and can be efficiently evaluated (see, e.g., \cite{gentleman1972}). With the coefficients at hand, a fast and stable evaluation of the resulting interpolant is available through, e.g., the Clenshaw algorithm \cite{fox1968chebyshev}. We also highlight that small distortions of function values at interpolation nodes only add small errors to the interpolation even for a high number of nodal points (see \cite[Section~2.3]{gass2018}).
This is important as it ensures the stability of the Chebyshev interpolation also when the function values at the interpolation nodes are calculated numerically, which is typically the case in our context.

The favourable properties outlined above together with the theoretical results of Corollary \ref{cor:cheb_conv} and Theorem \ref{thm:probErrBnd} motivate us to propose the Chebyshev acceleration of full re-evaluation. In the one-dimensional case, this corresponds to inserting the interpolation operator $I_N$ in place of the approximation operator $\mathbf{I}$ in Scheme \ref{fig:process_flow}.

Essential for the analysis of Chebyshev interpolations is the Bernstein ellipse. We recall the definition of a Bernstein ellipse $\mathcal{E}(\left[-1, 1\right], \varrho)$ with parameter $\varrho > 1$ as the open region in the complex plane bounded by the ellipse with foci $\pm 1$ and semiminor and semimajor axis lengths summing to $\varrho$. By $\mathcal{E}(\left[a, b\right], \varrho)$ we denote the linear transformation of $\mathcal{E}(\left[-1, 1\right], \varrho)$ by a map $\tau $ as in \eqref{eq:domain_transform}, see also \cite{gass2018}.

The choice of the interpolation domain is important to ensure optimal performance of the accelerated method. A fixed interpolation domain allows for a full separation of the accelerated pricing and the risk tasks. The proof of Theorem~\ref{thm:probErrBnd} provides details on the asymptotic choice of the domain size, but in practice a more heuristic approach might be taken, for example ensuring that a high percentage of samples lie within the domain. 
This domain can be further reduced by examining the properties of the value function under consideration. For certain options, such as barrier options, the payoff yields a natural bound. In other cases, the asymptotic behaviour of the value functions can be exploited, see also Appendix \ref{sec:calc_met_int_dom_deg}.
As an alternative, the simulation-based approach allows for a bounded interpolation domain. Namely, realised risk factors can be used to determine a suitable bounded interval, as any finite sample is bounded. Such an approach still separates the original pricer from the risk tasks and the theory of Proposition~\ref{prop:uniform_bound_11} directly applies.

To conclude this section, we note that Chebyshev interpolation has already been successfully applied in a variety of different financial settings, including~\cite{glau19iv,glau2020,GRZELAK2022,pachon2018,potz2020,zeron2021}. Recent collaborations further show the interest from practitioners to incorporate efficient Chebyshev-based technologies, see, e.g., \cite{nasdaq2024}. This moreover emphasizes the need for thorough analysis considering the importance of reliable financial risk calculation for the overall economy.

\section{Examples of exposures and their empirical estimates}\label{sec:exp_mea_exa}\label{sec:exp_meas_est}

We note that the quantile-based formulation
\begin{equation}\label{eq:qbrm}
\rho_m(X) = \int_0^1 Q_X(u)m(du),
\end{equation}
where $m$ is a probability measure on $(0,1)$,
allows for simultaneous analysis of exposure measures used in different areas of counterparty credit risk management practice, see also~\cite{cont2010}.

With $Q_{\mathbf{x}}(u) = x^{\left(\lfloor nu \rfloor + 1\right)}$, where $x^{(k)}$ is the $k$-th smallest element of the set $\{x^i\}_{i=1}^n$, we have the empirical estimate
\begin{equation}\label{eq:rho_m_emp_est}
\widehat{\rho}_m(\mathbf{x}) = \int_0^1 Q_{\mathbf{x}}(u)m(du) = \int_0^1 x^{\left(\lfloor nu \rfloor + 1\right)} m(du) = \sum_{i = 1}^n w_{n,i} x^{(i)},
\end{equation}
where $w_{n,i} = m\left(\frac{i-1}{n},\frac{i}{n}\right]$ for $i=1,...,n-1,$ and $w_{n,n} = m\left(\frac{n-1}{n},1\right)$.
  Note that calculating \eqref{eq:rho_m_emp_est} in general requires sorting a sample, amounting to an effort scaling linearly in $n\ln n$.

The form \eqref{eq:qbrm} covers many of the exposure measures used in practice for counterparty credit risk management and from \eqref{eq:rho_m_emp_est} we obtain their empirical estimators.
Let $\mathbf{X} = (X_1, ..., X_n)$ be the random samples from the exposure $X$, i.e., the sequence of independent random variables all distributed as $X$, and $\mathbf{x} = \{x^i\}_{i=1}^n$ the data set representing observations of $\mathbf{X}$.

	\emph{Expected exposure (EE)}, for instance, corresponds to choosing $m$ as the uniform distribution over $(0, 1)$,
\begin{equation}\label{eq:ee}
\mathrm{EE}(X) =  \int_0^1 Q_X(u)du = \mathbb{E} \left[ X \right].
\end{equation}
The expected exposure (under a risk-neutral measure) plays an important role in the calculation of a credit valuation adjustment: an adjustment made to the theoretical risk-neutral value of the trade to account for a counterparty risk, as discussed in Section \ref{sec:cva_sens}.

The \emph{Empirical EE} is given by
\begin{equation*}\label{eq:empirical_EE}
\widehat{\mathrm{EE}}(\mathbf{x}) = \frac{1}{n} \sum_{i = 1}^n x^i.
\end{equation*}
By the law of large numbers, if $\mathbb{E}\left[X^2\right]<\infty$, $\widehat{\mathrm{EE}}(\mathbf{X})$ converges to $\mathrm{EE}(X)$ almost surely (and therefore also converges in probability) as $n \to \infty$, and by the central limit theorem we have
\begin{equation}\label{eq:EE_clt}
\sqrt{n}\left(\widehat{\mathrm{EE}}(\mathbf{X}) -  \mathrm{EE}(X)\right) \xrightarrow[]{D} \mathcal{N}\left(0, \sigma^2_{\mathrm{EE}} \right) \text{ as } n \to \infty,
\end{equation}
where $\sigma^2_{\mathrm{EE}} = \textrm{Var}\left[ X \right]$.

	\emph{Potential Future Exposure (PFE)} corresponds to the choice of Dirac measure for $m$, $m = \delta_{\alpha}$ for a fixed $\alpha \in (0,1)$, that is
\begin{equation}\label{eq:pfe}
\mathrm{PFE}_{\alpha} (X) = Q_X(\alpha).
\end{equation}
	Thus, the probability of the exposure exceeding $\mathrm{PFE}_{\alpha}$ is equal to $1-\alpha$. The potential future exposure is often used for monitoring that the exposure to a certain counteparty is below the prespecified credit limit, a fundamental approach in counterparty credit risk management. A conditional version of PFE is used as the amount of collateral to be posted when trading in the presence of initial margins, see, e.g., \cite{Andersen2017IM,Andersen2017MPOR}.
	
	Mathematically, PFE coincides with the famous market risk measure - Value at Risk (VaR). Limitations of VaR (see, e.g., \cite{acerbi2002}, for an early reference), particularly its inability to capture the tail losses beyond the selected confidence level, have recently been addressed in the market risk regulation. Namely, moving from VaR to Expected Shortfall, i.e., average loss beyond VaR, is one of the key changes introduced by the Fundamental Review of the Trading Book \cite{bcbs2014frtb}.

The \emph{Empirical PFE}  is given by
\begin{equation*}\label{eq:empirical_PFE}
\widehat{\mathrm{PFE}}_{\alpha}(\mathbf{x}) = x^{(\lfloor n\alpha \rfloor + 1)}.
\end{equation*}
From \cite{sterfling1980} we know that $\widehat{\mathrm{PFE}}_{\alpha}(\mathbf{X}) \to \mathrm{PFE}_{\alpha}(X)$ with probability 1 as $n \to \infty$, and, if $X$ has a density $f_{X}(\cdot)$ in the neighbourhood of $\mathrm{PFE}_{\alpha}(X)$ and $f_{X}(\mathrm{PFE}_{\alpha}(X))>0$,
\begin{equation}\label{eq:PFE_clt}
\sqrt{n}\left(\widehat{\mathrm{PFE}}_{\alpha}(\mathbf{X}) -  \mathrm{PFE}_{\alpha}(X)\right) \xrightarrow[]{D} \mathcal{N}\left(0, \sigma^2_{\mathrm{PFE}_{\alpha}} \right) \text{ as } n \to \infty,
\end{equation}
where $\sigma^2_{\mathrm{PFE}_{\alpha}} = \frac{\alpha(1-\alpha)}{f_{X}(\mathrm{PFE}_{\alpha}(X))^2}$.

Defficiencies of PFE in the counterparty risk setting,  have been discussed in \cite{kenyon2018}. Therein, basing counterparty trading limits on average exposure beyond PFE rather than PFE has been proposed.
	
\emph{Credit Expected Shortfall (CES)} corresponds to choosing $m$ as the uniform distribution over $(\alpha, 1)$, where $\alpha \in (0,1)$ is fixed,
\begin{equation}\label{eq:ces}
\mathrm{CES}_{\alpha}(X) = \frac{1}{1-\alpha} \int_{\alpha}^1 Q_X(u)du,
\end{equation}
	and hence accounts for average exposure beyond the PFE level. We observe that some major clearing houses already base their margining methodologies on (conditional) CES.	

The     \emph{Empirical CES}  is given by
\begin{equation*}\label{eq:empirical_CES}
\widehat{\mathrm{CES}}_{\alpha}(\mathbf{x}) = \frac{1}{n(1-\alpha)}\left( x^{(\lfloor n\alpha \rfloor + 1)}(\lfloor n\alpha \rfloor + 1 - n\alpha) + \sum_{i=\lfloor n\alpha \rfloor + 2}^n x^{(i)} \right).
\end{equation*}
From \cite{trinidade2007} we have that when $\mathbb{E}\left[X^2\right]<\infty$, $\widehat{\mathrm{CES}}_{\alpha}(\mathbf{X})$ is a consistent estimator of $\mathrm{CES}_{\alpha}(X)$, i.e., $\widehat{\mathrm{CES}}_{\alpha}(\mathbf{X}) \to \mathrm{CES}_{\alpha}(X)$ in probability as $n \to \infty$. Furthermore, if $X$ has a density $f_{X}(\cdot)$ in the neighborhood of $\mathrm{PFE}_{\alpha}(X)$, then
\begin{equation}\label{eq:CES_clt}
\sqrt{n}\left(\widehat{\mathrm{CES}}_{\alpha}(\mathbf{X}) -  \mathrm{CES}_{\alpha}(X) \right) \xrightarrow[]{D} \mathcal{N}\left(0, \sigma_{\mathrm{CES}_{\alpha}}^2 \right) \text{ as } n \to \infty,
\end{equation}
where
\begin{equation*}
\sigma_{\mathrm{CES}_{\alpha}}^2 = \lim_{n \to \infty} n \textrm{Var}\left[ \widehat{\mathrm{CES}}_{\alpha}(\mathbf{X}) \right] = \frac{1}{(1-\alpha)^2} \textrm{Var}\left[ \left[X - \mathrm{PFE}_{\alpha}(X) \right] \cdot \mathbbm{1}_{\left\{X>\mathrm{PFE}_{\alpha}(X)\right\}} \right].
\end{equation*}

Let us also mention here that \eqref{eq:qbrm} includes \emph{spectral exposure measures (SEM)} \cite{ACERBI2002srm} which generalize CES and correspond to choosing $m(du) = \phi(u)du$, where weight function $\phi:[0,1] \to [0, \infty)$ is an increasing density on $[0,1]$,
\begin{equation}\label{eq:sem}
\rho_{\phi}(X) = \int_0^1 Q_X(u)\phi(u)du.
\end{equation}
While $\mathrm{CES}_{\alpha}$ assigns weights to all PFEs beyond the confidence level $\alpha$ equally, SEM allow these weights to be freely chosen. Increasing weight functions imply that marginal costs are increasing with exposure.

\section{Deferred Proofs and Calculations}\label{app:1l_q}
\subsection{Proof of Theorem~\texorpdfstring{\ref{eq:peb}}{2.1}}\label{app:probErrBoundProof}

For $\kappa>0$, set $L =\left( \left(\ln{n} + \ln{\alpha} +  \kappa^2/\Large(18\sigma_\rho^2\Large)\right)/\beta \right)^{1/\gamma}$, $N = \left\lceil L^{{\vartheta}} - \ln\left(\kappa/(3a\sqrt{n})\right)/b\right\rceil^{{D}}$, and $M = \left\lceil n \mathfrak{c}^2 (1+\ln{N})^2\sigma_\text{ref}^2 \left(18\ln(N)/\kappa^2 + \sigma_\rho^{-2} \right)  \right\rceil$.  Note that we then have $L^{{\vartheta}} - {\sqrt[D]{N}} \leq \ln\left(\kappa/(3a\sqrt{n})\right)/b$.

We first show the probabilistic estimate. 
Consider the event $F$ of all samples lying within $\Omega_L$, and the error of the full re-evaluation and the reference pricer being bounded by $\kappa/(3\sqrt{n})$. Namely,  
$F = A\cap B \cap C$, where $A = \cap_{i=1}^n\left\{Z_t^i \in \Omega_L\right\}$, $B = \left\{\vert \rho(X_t) - \widehat{\rho}(\mathbf{X}_t) \vert \leq \kappa/(3\sqrt{n})\right\}$,  $C=  \big\{\max\limits_{1 \leq j \leq N} |V_t(\zeta_t^j) - V_{t}^M(\zeta_t^j)| \leq \kappa/(3\sqrt{n}\mathfrak{c} (1+\ln N)) \big\}$.

In the event F, we can apply the overall error split \eqref{eq:err_split}, which together with assumptions (A3) and (A4) yields
\begin{align*}
\left\vert \rho(X_t) - \widehat{\rho}(\mathbf{Y}_t)\right\vert &\leq \left\vert \rho(X_t) - \widehat{\rho}(\mathbf{X}_t)\right\vert +\Vert V_t{^+} - (I_N V_{t}^{M}){^+} \Vert_{\infty,\Omega_L}\\
&\leq\left\vert \rho(X_t) - \widehat{\rho}(\mathbf{X}_t)\right\vert +\Vert V_t - I_N V_{t} \Vert_{\infty,\Omega_L}+\Vert I_N V_{t} - I_N V_{t}^{M} \Vert_{\infty,\Omega_L}\\
& \leq  \frac{\kappa}{3\sqrt{n}}  + a e^{b(L^{\vartheta} - {\sqrt[D]{N}})}+ \frac{\kappa}{3\sqrt{n}} \leq  \frac{\kappa}{\sqrt{n}}.
\end{align*}

To estimate the probability of the event $F$, we first note that $\mathbb{P}(F) \geq \mathbb{P}(A) + \mathbb{P}(B) + \mathbb{P}(C) - 2$.
Under assumption (A2), for $n$ large enough, we have
\begin{equation*}
\mathbb{P}(B)   \geq 1-n\mathbb{P}(Z_t\notin I) \geq 1-n \alpha e^{-\beta L^\gamma}
 = 1 - e^{-\frac{\kappa^2}{18\sigma_\rho^2}}
\end{equation*}
while by assumption (A1)
\begin{align*}
\mathbb{P}(C) = 2 \Phi\left(\frac{\kappa}{3\sigma_\rho}\right) - 1
\geq 1 - e^{-\frac{\kappa^2}{18\sigma_\rho^2}} 
\end{align*}
where $\Phi$ is the distribution function of the standard normal distribution.
We also have by (A1)
\begin{align*}
\mathbb{P}(D) & \geq 1-\sum_{j=1}^N \mathbb{P}\left( |V_t(\zeta_t^j) - V_{t}^M(\zeta_t^j)| > \frac{\kappa}{3\sqrt{n}\mathfrak{c} (1+\ln N)}\right)
\\
&\geq
1-N e^{- \frac{\kappa^2  M}{18 n \mathfrak{c}^2 (1+\ln N)^2 \bar{\sigma}^2} }\geq 1 - e^{-\frac{\kappa^2}{18\sigma_\rho^2}},
\end{align*}
 which completes the proof of the probabilistic estimate.
 
Under Assumption (A5), we estimate the effort to construct the function approximation and evaluate $\widehat{\rho}(\mathbf{Y_t})$ as follows. We combine the effort to:
\begin{enumerate}
\item Evaluate $V_{t}^M$ at $N=\left\lceil \left( \frac{\ln{n} + \ln{\alpha} +  \frac{\kappa^2}{18\sigma_\rho^2}}{\beta} \right)^{{{\vartheta}}/\gamma} - \frac{\ln\left(\frac{\kappa}{3a\sqrt{n}}\right)}{b}\right\rceil^{{D}}$  points, the effort of each evaluation scaling with $M = \left\lceil n \mathfrak{c}^2 (1+\ln{N})^2\bar{\sigma}^2 \left(\frac{18\ln(N)}{\kappa^2} + \sigma_\rho^{-2} \right)  \right\rceil$,
\item Construct $\mathbf{I}_N V_t^M $, the effort of which scales with $N^{{\xi}} $, and 
\item Evaluate $\mathbf{I}_N V_t^M $ at $n$ points, the effort of each evaluation scaling with $N$.
\end{enumerate}
The first item is dominating the effort and yields an effort scaling with no more than $\Lambda(n) = n \ln{n}^{{D}\max\{1, {{\vartheta}}/\gamma\}} (\ln^2{n})^3 \leq c n^{1+\varepsilon}$ for any $\varepsilon>0$. 

\subsection{Class of pricing problems fulfilling Assumption~\texorpdfstring{\eqref{ass2}}{(A3)}}\label{sec:ass_a3}

We will assume that the asset prices evolve as exponential semimartingales. Let the driving process be an $\mathbb{R}^d$-valued semimartingale $X_t = (X^1_t, \dots, X^d_t)^\top$ and $S_t = (S_t^1, \dots, S_t^d)^\top$ be the vector of asset price processes; then each component $S_t^i$ of $S_t$ is modelled as an exponential semimartingale, i.e.
\begin{equation*}
S^i_t = S^i_0 \exp X^i_t, \quad 0 \leq t \leq T, \quad 1 \leq i \leq d,
\end{equation*}
where $X^i$ is an $\mathbb{R}$-valued semimartingale.

We will price options with payoff $\Phi(S_0e^{X_T}) = f(X_T - s)$ at maturity $T$, and assume $f$ is a measurable function $f: \mathbb{R}^d \to \mathbb{R}_+$, and $s = (s^1, \dots, s^d) \in \mathbb{R}^d$ with $s^i = - \log S^i_0$. Furthermore, we assume for simplicity that interest rate and dividend yield are zero. Then, time-$0$ price of an option is given by $V(s) = V_f(s;X) = E[f(X_T - s)]$. We will also make use of the dampened payoff function,
\begin{equation*}
g(x) := e^{-\langle R,x \rangle} f(x) \quad \text{for } x \in \mathbb{R}^d,
\end{equation*}
where $R \in \R^D$ serves as a dampening coefficient.

We will make the following assumptions, which are adapted from~\cite[Assumptions (A1)--(A4)]{gass2018} to the present setting.

\noindent\textbf{Conditions B:}
\begin{enumerate}
	\item[\textbf{(B1)}] $g \in L^1(\R^D)$,
	\item[\textbf{(B2)}] $\widehat{e^{\langle R,\cdot\rangle} P_{X_T}} \in L^1(\R^D)$,
	\item[\textbf{(B3)}] $M_{X_T}(R) = \mathbb{E}(e^{RX_T})< \infty$, and
	\item[\textbf{(B4)}] There exist constants $c_1,c_2>0$ and $\alpha \in (1,2]$ such that $|\phi_{X_T}(u-iR)| \leq c_1 e^{-c_2\Vert u  \Vert_2^{\alpha}}$ for all $u \in \R^D$, where $\phi_{X_T}(u) = \mathbb{E}(e^{i\langle u,X_T\rangle})$ is the characteristic function of $X_T$. 
\end{enumerate}

Conditions (B1)--(B4) are typically satisfied, and a discussion of various models and payoffs that fit this framework can be found in~\cite[Chapter 4.3]{gass2018}.

\begin{lemma}\label{lem:fulfilling_assumption_a3}
If Conditions (B1)--(B4) are in force, our condition~\eqref{ass2} on the domain size dependence is satisfied for the standard multivariate Chebyshev interpolation with a $N$  interpolation points on a hypercube domain $\Omega_L =  [-L,L]^D$, i.e.
\[
\hspace{0.83cm}\Vert V_T - \mathbf{I}_N V_T \Vert_{L^{\infty}(\Omega_L)}\leq a e^{b(L^{\vartheta} - \sqrt[D]{N})} 
%\hspace{0.15cm}
\text{ with }
%\hspace{0.15cm}
a,b>0, \text{ and }\vartheta = \frac{\alpha}{\alpha-1} \geq 2. 
\]
\end{lemma}
\begin{proof}
Under the assumptions, the option pricing function $V:\Omega_L\to\R$ can be analytically extended to a generalized Bernstein ellipse $B(\Omega_L, \varrho)$, see~\cite[Theorem 3.2]{gass2018}. The extension is given by the Fourier transform valuation formula, i.e.,
\begin{equation}\label{eq:valuation}
V(z) = \frac{e^{-\langle R,z\rangle}}{(2\pi)^D} \int_{\R^D} e^{-i\langle u,z\rangle}\phi_{X_T}(u-iR)\widehat{f}(iR-u)du,
\end{equation}
for $z \in B(\Omega_L, \varrho)$. Here $\hat{f}(u) = \int_{\mathbb{R}^D} e^{i\langle u, x\rangle}f(x)dx $ is the Fourier transform of $f$.
As $g \in L^1(\R^D)$  by (B1), the Riemann-Lebesgue lemma yields uniform boundedness of $\widehat{f}(iR-u)$ in $u$, and in particular we have
\begin{equation}
|\widehat{f}(iR-u)| = |\widehat{g}(-u)| \leq \Vert g \Vert_{L^1(\R^D)} < \infty, \quad \forall u \in \R^D.
\end{equation}
Together with the assumption (B4) this yields the bound
\begin{equation}
|V(z)| \leq \frac{c_1 e^{-\langle R,\Re(z) \rangle}}{(2\pi)^D} \|g\|_{L^1(\mathbb{R}^D)}\int_{\R^D}  e^{-c_2 \Vert u\Vert_2^{\alpha} + \langle u,\Im(z)\rangle}du.
\end{equation}
In the following we put $x = \Re(z)$ and $y = \Im(z)$.

We next estimate the following integral,
\begin{equation}
I = \int_{\R^D}  e^{-c_2 \Vert u\Vert_2^{\alpha} + \langle u,y\rangle}du.
\end{equation}
Employing Cauchy-Schwarz then Young inequality, we see
\begin{equation}
\langle u, y\rangle \leq \Vert u \Vert_2 \Vert y \Vert_2 \leq \frac{\lambda}{\alpha} \Vert u \Vert_2^\alpha + \frac{\alpha-1}{\alpha}\lambda^{-\frac{1}{\alpha-1}} \Vert y \Vert_2^{\frac{\alpha}{\alpha-1}},
\end{equation}
for any $\lambda>0$, and so
\begin{equation}
I \leq e^{\gamma_1 \Vert y \Vert_2^{\frac{\alpha}{\alpha-1}}}\int_{\R^D}  e^{-\gamma_2 \Vert u\Vert_2^{\alpha}}du = 2\pi^{\frac{D}{2}} \frac{\Gamma(D/\alpha)}{\Gamma(D/2)} \gamma_2^{1-\frac{D}{\alpha}}e^{\gamma_1 \Vert y \Vert_2^{\frac{\alpha}{\alpha-1}}}
\end{equation}
where $\gamma_1 = (\alpha-1)/(\alpha\lambda^{1/(\alpha-1)})$ and we choose $\lambda$ so that $\gamma_2 = c_2 - \frac{\lambda}{\alpha}>0$.
The equality follows from a direct calculation of the integral with the change to spherical coordinates for the radial integrand.

Combining the above calculations, we obtain overall
\begin{equation}
|V(z)| \leq \frac{c_1 \gamma_2^{1-\frac{D}{\alpha}}}{2^{D-1}\pi^{D/2}} \frac{\Gamma(D/\alpha)}{\Gamma(D/2)} e^{\gamma_1 \Vert y \Vert_2^{\frac{\alpha}{\alpha-1}} -\langle R,x \rangle}.
\end{equation}

Now, for $z = x + iy \in B(\Omega_L, \varrho)$, we first note $y_i^2 \leq (\xmax_i-\xmin_i)^2(\varrho_i - \varrho_i^{-1})^2/16 \leq L^2(\varrho_i - \varrho_i^{-1})^2/4$, so that we have an upper bound
\begin{equation}
\Vert y \Vert_2^{\frac{\alpha}{\alpha-1}} \leq \left(\sum_i^{D} (\varrho_i - \varrho_i^{-1})^2/4 L^2 \right)^{\frac{\alpha}{2(\alpha-1)}} = a_1 L^{\frac{\alpha}{\alpha-1}}.
\end{equation}
for $a_1 = \left(\sum_i^{D} (\varrho_i - \varrho_i^{-1})^2/4 \right)^{\frac{\alpha}{2(\alpha-1)}} $.
%where we put $a_i^2 = (\varrho_i - \varrho_i^{-1})^2/4$ and $L_i = \xmax_i-\xmin_i$. 
Similarly, $x_i^2 \leq  (\xmax_i-\xmin_i)^2(\varrho_i + \varrho_i^{-1})^2/16 \leq L^2 (\varrho_i + \varrho_i^{-1})^2/4$, i.e., $|x_i|\leq L (\varrho_i + \varrho_i^{-1})/2$, and so
\begin{equation}
\langle R, x \rangle \geq -\sum_i^D |R_i|L (\varrho_i + \varrho_i^{-1})/2 =  -a_2  L,
\end{equation}
for $a_2=\sum_i^D |R_i| (\varrho_i + \varrho_i^{-1})/2$.

Finally, we have
\begin{equation}
|V(z)| \leq \Vert g\Vert_{L^1(\R^D)}\frac{c_1 \gamma_2^{1-\frac{D}{\alpha}}}{2^{D-1}\pi^{D/2}} \frac{\Gamma(D/\alpha)}{\Gamma(D/2)} e^{\gamma_1 a_1 L^{\frac{\alpha}{\alpha-1}} + a_2L} \leq c_3e^{c_4L^{\frac{\alpha}{\alpha-1}}}.
\end{equation}

The final results follows from standard Chebyshev convergence results, e.g.~\cite[Corollary 2]{gass2018}.
\end{proof}

\subsection{Proof of Proposition \ref{prop:emp_lp}}\label{app:proof_emp_lp}

For a vector $v=(v_i)_{i=1}^{n} \in\mathbb{R}^n$, we consider the $\ell^p$ norm ($1\leq p<\infty$) as $\|v\|_{\ell^p} = \left(\sum_{i=1}^n |v_i|^p\right)^{1/p}$.

(a) For any $c> 0$, $S_c = \{z \in \mathbb{R} : |V_t{^+}(z) - U_t{^+}(z)| > c\}$, $\lambda(S_c) \leq \frac{\Vert V_t{^+} - U_t{^+} \Vert_{L^p(\mathbb{R})}^p}{c^p} \leq \frac{\Vert V_t - U_t \Vert_{L^p(\mathbb{R})}^p}{c^p}$, where $\lambda$ is the Lebesgue measure, and so
\begin{equation*}
\mathbb{P}\left(|V_t(Z_t){^+}  - U_t(Z_t){^+} |>c\right) = \mathbb{P}(Z_t \in S_c) = \int_{S_c} f_{Z_t}(z)dz \leq \frac{\Vert f_{Z_t} \Vert_{\infty}\Vert V_t - U_t \Vert_{L^p(\mathbb{R})}^p}{c^p},
\end{equation*}
and, using the union bound:
\begin{equation*}
\mathbb{P}\left(\bigcup_{i=1}^n \{ |V_t(Z_t^i){^+}  - U_t(Z_t^i){^+} |>c\}\right)\leq n \frac{\Vert f_{Z_t} \Vert_{\infty}\Vert V_t - U_t \Vert_{L^p(\mathbb{R})}^p}{c^p}.
\end{equation*}
Thus,
\begin{align*}
\mathbb{P}\left(\max_{i=1,\ldots,n} |X_t^i - Y_t^i| \leq  c\right) &= 1-\mathbb{P}\left(\bigcup_{i=1}^n\{|V_t(Z_t^i){^+}  - U_t(Z_t^i){^+} |>c\}\right) \\&\geq 1 -n \frac{\Vert f_{Z_t} \Vert_{\infty}\Vert V_t - U_t \Vert_{L^p(\mathbb{R})}^p}{c^p}.
\end{align*}
Taking $c=\left(\frac{\|f_{Z_t}\|_\infty}{\eta}\right)^{1/p} n^{1/p} \Vert V_t\ - U_t \Vert_{L^p(\mathbb{R})}$ for $\eta \in (0, 1]$, we have
\begin{align*}
\mathbb{P}\left(\max_{i=1,\ldots,n} |X_t^i - Y_t^i| \leq  \left(\frac{\|f_{Z_t}\|_\infty}{\eta}\right)^{1/p} n^{1/p}\Vert V_t - U_t \Vert_{L^p(\mathbb{R})}\right) \geq 1-\eta.
\end{align*}
By the proof of Proposition \ref{prop:uniform_bound_11} we conclude that 
\begin{equation*}
\left|\widehat{\rho}(\mathbf{X}_t) - \widehat{\rho}(\mathbf{Y}_t)\right| \leq \left(\frac{\|f_{Z_t}\|_\infty}{\eta}\right)^{1/p} n^{1/p} \|V_t - U_t\|_{L^p(\mathbb{R})}
\end{equation*}
with probability at least $1-\eta$.

(b) {Given the density $f_m$, the empirical quantile-based exposure measure~\eqref{eq:rho_m_emp_est} is given as
\[
\widehat{\rho}_m(\mathbf{X}_t) = \sum_{i=1}^n w_{n, i} X_t^{(i)},
\]
with $w_{n, i} = \int_{\frac{i-1}{n}}^{\frac{i}{n}} f_m(u)\mathrm{d}u$. Thus the sample error is 
\begin{align*}
|\widehat{\rho}_m(\mathbf{X}_t) - \widehat{\rho}_m(\mathbf{Y}_t)| \leq
\sum_{i=1}^n w_{n, i} |X_t^{(i)}-Y_t^{(i)}|
\leq \|( w_{n, i})_i\|_q \|(X_t^{(i)}-Y_t^{(i)})_i\|_{\ell^p}
\end{align*}
where we have applied H\"older's inequality for $1/p + 1/q = 1$.
Lemma~\ref{lem:ord_diffs} shows that the difference of two ordered sequences is no larger than the difference of two unordered sequence:
\begin{align*}
\|(\mathbf{X}_t^{(i)}-\mathbf{Y}_t^{(i)})_i\|_{\ell^p} \leq  \|\mathbf{X}_t-\mathbf{Y}_t\|_{\ell^p}.
\end{align*}

Furthermore, we can apply H\"older's inequality on the product $1\cdot f(x)$ to estimate 
\begin{align*}
\|( w_{n,i })_i\|_q^q &= \sum_{i=1}^n\left(\int_{\frac{i-1}{n}}^{\frac{i}{n}} f_m(u)du\right)^q \leq \sum_{i=1}^n\int_{\frac{i-1}{n}}^{\frac{i}{n}} f_m(u)^q du n^{-q/p} = n^{-q/p}  \int_0^1 f_m(u)^q du,
\end{align*}
i.e. $\|( w_{i,n})_i\|_q \leq n^{-1/p} \|f_m\|_{L^q}$. 

Putting both estimates together, we have 
\[
|\widehat{\rho}_m(\mathbf{X}_t) - \widehat{\rho}_m(\mathbf{Y}_t)| \leq
\|f_m\|_{L^q} \, \frac{1}{\sqrt[p]{n}}\|\mathbf{X}_t-\mathbf{Y}_t\|_{\ell^p}.
\]
For $\eta \in (0,1]$, by applying  Hoeffding's inequality, we obtain
\begin{align*}\label{eq:hoeff_bound}
\frac{1}{n} \Vert \mathbf{X}_t - \mathbf{Y}_t \Vert_{\ell^p}^p &\leq \Vert \mathbf{X}_t - \mathbf{Y}_t \Vert_{L^p(\Omega)}^p +  \sqrt{\frac{\ln(1/\eta)}{n}} \Vert V_t{^+} - U_t{^+}\Vert_{\infty}^p %\\
%&\leq \Vert X_t - Y_t \Vert_{L^p(\Omega)}^p +  \sqrt{\frac{\ln(1/\eta)}{n}} \Vert V_t - U_t\Vert_{\infty}^p,
\end{align*}
with probability at least $1 - \eta$, from which
\begin{equation*}
\begin{split}
|\widehat{\rho}_m(\mathbf{X}_t) - \widehat{\rho}_m(\mathbf{Y}_t)|  &\leq \|f_m\|_{L^q} \Vert \mathbf{X}_t - \mathbf{Y}_t \Vert_{L^p(\Omega)} + \|f_m\|_{L^q}  \sqrt[2p]{\frac{\ln(1/\eta)}{n}}\Vert V_t^+ - U_t^+\Vert_{\infty}\\
&\leq \|f_m\|_{L^q}\Vert f_{Z_t} \Vert_{\infty}^{\frac{1}{p}}\Vert V_t - U_t \Vert_{L^p(\mathbb{R})} +  \|f_m\|_{L^q}\sqrt[2p]{\frac{\ln(1/\eta)}{n}} \Vert V_t - U_t\Vert_{\infty}.
\end{split}
\end{equation*}
}

\subsection{Statement and proof of Lemma \ref{lem:ord_diffs}}\label{}
\begin{lemma}\label{lem:ord_diffs}
 For two sequences of real numbers $\mathbf{a} = \left(a^i\right)_{i=1}^{n}$ and $\mathbf{b} = \left(b^i\right)_{i=1}^{n}$, consider $\mathbf{a}^{\uparrow}  = (a^{(i)})_{i=1}^n$ and $\mathbf{b}^{\uparrow}  = (b^{(i)})_{i=1}^n$, where  $a^{(i)}$ and $b^{(i)}$ denote the $i$-th smallest elements of the sets $\{a^i\}_{i=1}^n$ and $\{b^i\}_{i=1}^n$ respectively. 
 For any $p\in[1,\infty]$,we have
\[\left\Vert \mathbf{a}^{\uparrow} - \mathbf{b}^{\uparrow} \right\Vert_{\ell^p} \leq \Vert \mathbf{a} - \mathbf{b} \Vert_{\ell^p} \]
\end{lemma}

The lemma can be shown using the theory of majorization, see, e.g.~\cite{marshall2011}. We start considering $p\in[1, \infty)$. By~\cite[Chapter 6,  Proposition A.2]{marshall2011}, $\mathbf{a}-\mathbf{b}$ majorizes $\mathbf{a}^\uparrow - \mathbf{b}^\uparrow$. As $f(x) = |x|^p$ is continuous and convex,~\cite[Chapter 4, Proposition B.1]{marshall2011} shows that
\[ \sum_{i=1}^n |a^{(i)} - b^{(i)}|^p \leq \sum_{i=1}^n |a^{i} - b^{i}|^p, \]
hence  $\|\mathbf{a}^\uparrow - \mathbf{b}^\uparrow\|_{\ell^p} \leq \|\mathbf{a} - \mathbf{b}\|_{\ell^p}$. Taking the limit  $p\rightarrow \infty$ shows the statement for all $p\in[1, \infty]$.

\section{Calculation methodology of the numerical experiments}
In this chapter, we first introduce the option pricing models and their  calibration and then provide a step-by-step algorithm for the numerical experiments. 
\subsection{Option pricing models and calibration to market data}\label{app:params}

	We recall considered pricing models and options in Table \ref{tab:parametrisation}, wherein we also report parameters used in the numerical experiments. In the following we describe the procedure employed for calibration of the pricing models to market data.
	
	\begin{table}[]
\centering
\begin{tabular*}{\linewidth}{@{}lr@{\extracolsep{4pt}}rr@{}}
\toprule
                                                                                                                                                                                                  & \multicolumn{3}{c}{\textbf{Parameters}}                                                                                                                               \\ \cmidrule{2-2} \cmidrule{3-4}
\textbf{Model}                                                                                                                                                                                    & Market-observed                                                    & Calibrated                                         & \textbf{ARPE}                               \\ \midrule
Black-Scholes-Merton (BSM) \cite{Black1973,Merton1973}                                                                                                                                                                        & $S_0 = 3825.33$                                                    & $\sigma = 0.1943$                                  & $12.55\%$                                   \\
$dS_t = \mu S_t dt + \sigma S_t dW_t$                                                                                                                                                             & $\mu = 0.11$                                                       &                                                    &                                             \\
                                                                                                                                                                                                  & $r = 0.0110$                                                       &                                                    &                                             \\
                                                                                                                                                                                                  & \multicolumn{1}{l}{}                                               &                                                    &                                             \\
Merton's jump-diffusion (MJD) \cite{Merton1976}                                                                                                                                                                     & $S_0 = 3825.33$                                                    & $\sigma = 0.1483$                                  & $6.97\%$                                    \\
\multirow{2}{*}{$\begin{aligned} dS_t &= \mu S_t dt + \sigma S_t d W_t + J_t S_t d Q_t\\ J_t &\sim \mathcal{N}(\gamma,\delta^2)\end{aligned}$}                                                   & $\mu = 0.11$                                                       & $\lambda = 1.2998$                                 &                                             \\
                                                                                                                                                                                                  & $r = 0.0110$                                                       & $\gamma = -0.1475$                                 &                                             \\
                                                                                                                                                                                                  & \multicolumn{1}{l}{}                                               & $\delta  = 0.1331$                                 &                                             \\
                                                                                                                                                                                                  & \multicolumn{1}{l}{}                                               &                                                    &                                             \\
Heston's stochastic volatility (HSV)                                                                                                                                                             \cite{Heston1993} & $S_0 = 3825.33$                                                    & $\nu_0 = 0.0650$                                   & $2.88\%$                                    \\
\multirow{4}{*}{$\begin{aligned} d S_t & = \mu S_t dt + \sqrt{\nu_t} S_t dW_t^S \\ d \nu_t & = \kappa (\theta - \nu_t)dt + \eta \sqrt{\nu_t}dW_t^{\nu}\\ d&W^S dW^{\nu} = \rho dt \end{aligned}$} & $\mu = 0.11$                                                       & $\kappa = 2.3134$                                  &                                             \\
                                                                                                                                                                                                  & $r = 0.0110$                                                       & $\theta = 0.0889$                                  &                                             \\
                                                                                                                                                                                                  &                                                                    & $\eta = 0.9898$                                    &                                             \\
                                                                                                                                                                                                  &                                                                    & $\rho = -0.7040$                                   &                                             \\ \midrule
\textbf{Option}                                                                                                                                                                                   & \multicolumn{1}{l}{\textbf{Payoff}}                                & \multicolumn{2}{r}{\textbf{Parameters}}                                                          \\ \midrule
European call                                                                                                                                                                                     & \multicolumn{1}{l}{$(S_T-K)^+$}                                    & \multicolumn{2}{r}{\multirow{4}{*}{$\begin{aligned}T=1,&\\ K=3825.33,&\\ B =5738& \end{aligned}$}} \\
Digital put                                                                                                                                                                                       & \multicolumn{1}{l}{$\mathbbm{1}_{\{S_T < K\}}$}                    & \multicolumn{2}{r}{}                                                                             \\
Up-and-out barrier call                                                                                                                                                                           & \multicolumn{1}{l}{$(S_T-K)^+ \mathbbm{1}_{\{\max_{t\leq T} S_t < B\}}$} & \multicolumn{2}{r}{}                                                                             \\
American put                                                                                                                                                                                      & \multicolumn{1}{l}{$(K-S_t)^+$ if exercised at $t\leq T$}                & \multicolumn{2}{r}{}                                                                             \\ \bottomrule
\end{tabular*}
\caption{Parametrisation of options and models with the model fit to market data measured by the calibration error ARPE \eqref{eq:arpe}. $S_t$ denotes an uncertain price of an underlying asset at time $t>0$. $W, W^S$, and $W^{\nu}$ are standard Brownian motions, where $W^S$ and $W^{\nu}$ have correlation $\rho$. $Q$ is a Poisson process with intensity $\lambda>0$ .}
\label{tab:parametrisation}
\end{table}

	Aiming to match the observed market prices as closely as possible, we estimate the parameters of the pricing models by minimizing average relative pricing errors (ARPE),
	\begin{equation}\label{eq:arpe}
	\mathrm{ARPE} = \frac{1}{\text{number of options}} \sum_{\text{options}} \frac{\left|\text{market price } - \text{model price}\right|}{\text{market price}},
	\end{equation}
	using the differential evolution algorithm \cite{Storn1997} (see \cite{Gilli2012} for the application to calibration). The choice of $\mathrm{ARPE}$ for the objective function is motivated by the fact that we include options with a wide range of strikes, leading to a both very small and very large option prices in the dataset.
	
	The market data obtained from Barchart.com consists of prices of European options written on S\&P 500 Index (SPX) traded on July 1, 2022, when the index had a value of $S_0 = 3\,825.33$. We consider the options with 1, 2, 3, 6, 9, and 12 months to expiry, and all available strikes. We convert the prices of put options to the prices of the corresponding call options using the put-call parity. We set the risk-free rate to $r = 0.0110$, which is the yield of a US treasury one year zero coupon bond on July 1, 2022 ($0.0279$) reduced by the dividend yield of SPX reported on June 30, 2022 ($0.0169$). We estimate the drift under the historical measure as a simple average of yearly returns in the past 10 years, yielding $\mu = 0.11$. We report calibrated model parameters and the corresponding $\mathrm{ARPE}$s in Table \ref{tab:parametrisation}.

We simulate the scenarios under the physical measure, i.e., using the historical average return $\mu$ as the drift, while for pricing we use a risk-neutral measure by adapting the drift accordingly.

\subsection{Step-by-step methodology of the numerical experiments}\label{sec:calc_met_int_dom_deg}

In this section, we provide implementation details of the numerical experiments presented in Section \ref{subsec:num_res}. We highlight the key steps of the computations.

As a first step, we simulate three different sets of market scenarios, one for each of the calibrated models from Table \ref{tab:parametrisation}. Once simulated, scenarios remain fixed throughout the experiments.

\emph{1. Simulate $n$ paths $\{z_{t_u}^i : u = 0,1,..., m\}$, $i = 1,...,n$, of the risk factors over $m+1$ equidistant times $t_u$ in $[0, T]$ using the Euler-Maruyama method.}

	In the BSM and MJD models, $z_{t_u}^i = s_{t_u}^i$, while in the HSV model, $z_{t_u}^i = \left(s_{t_u}^i, \nu_{t_u}^i\right)$. For the experiments of Section \ref{subsec:num_res} we chose $n = 10\,000$ and $m = 52$.

\emph{2. Evaluate reference pricers $V_{t_u}$ path-wise, i.e., calculate $x_{t_u}^i = V_{t_u}(z_{t_u}^i)^+$ for all $i = 1,...,n$ and $u = 1,...,m$.}

We decide to evaluate the reference pricers path-wise in a loop to keep the different methods comparable.

\emph{3. Accelerate the path-wise exposure calculations with the Chebyshev interpolation.}

\emph{3.1. Set-up the interpolation domain for every time point $t_u$.}
		
		With realisations of the risk factors at hand, the interpolation domain at $t_u$ in the asset-price dimension can be simply chosen as $\Large[\underline{s}, \overline{s}\Large]$, where $\underline{s} = \min_{i = 1,...,n} \{ s_{t_u}^i \}$, and $\overline{s} = \max_{i = 1,...,n} \{ s_{t_u}^i \}$. A more informed choice of the interpolation domain can be made based on the explicitly available formulas for the asymptotic behaviour of the considered options' value functions. Namely, $V_{t_u}$ behaves as either appropriately discounted payoff or $s \mapsto 0$ for sufficiently high or low levels of underlying's price $s$.
		
		As approximating the analytically known asymptotics  is an unnecessary effort, we (attempt to) truncate the interval $[\underline{s}, \overline{s}]$ by the bisection method. That is, we search for the points $\underline{s}^{\prime} \geq \underline{s}$ and $\overline{s}^{\prime} \leq \overline{s}$ where the value function $V_{t_u}$ is sufficiently close to the analytically known asymptotic formulas. The interpolation is then done on $[\underline{s}^{\prime}, \overline{s}^{\prime}]$, the resulting polynomial approximation called for $s_{t_u}^i \in [\underline{s}^{\prime}, \overline{s}^{\prime}]$, and known analytical formulas used for evaluating at $s_{t_u}^i \in [\underline{s}, \overline{s}] \setminus [\underline{s}^{\prime}, \overline{s}^{\prime}]$.

For the HSV model, we use the rectangular interpolation domain induced by the MC paths, obtained by taking the Cartesian product of previously defined $[\underline{s}, \overline{s}]$ with the analogously defined interval in the variance dimension, $[\underline{\nu}, \overline{\nu}]$.

Additionally, following \cite{glau2019fast}, we split the interpolation domain in the asset-price dimension at strike $K$ in all cases. The smoothness of the value functions usually deteriorates around $K$ hindering the speed of convergence of the Chebyshev interpolation. The effect is especially pronounced for short maturities where $V_{t_u}$ approaches a payoff function which either has a kink or a jump at $K$.

\emph{3.2. Construct the interpolation of the pricer $z \mapsto V_{t_u}(z)$ for every $t_u$, i.e., obtain the polynomial approximations $U_{t_u}(z) \approx V_{t_u}(z) $, $u = 1,...,m$.}

\emph{3.3. Evaluate obtained approximation $U_{t_u}$ path-wise, i.e., calculate $y_{t_u}^i = U_{t_u}(z_{t_u}^i)^+$ for $i = 1,...,n$, $u = 1,...,m$.}

To maintain comparability, as in Step 2. we evaluate Chebyshev approximations path-wise in a loop.

\emph{4. Addressing knock-outs and early-exercises for path-dependent options.}
	
		When calculating the exposures of the barrier options, if the option is knocked-out on path $i$ at time $t_u$, set the exposure at that and all future time steps on path $i$ to zero.

		When calculating the exposures of a American option, we first estimate the early-exercise boundary for each time-step. If the option is exercised on path $i$ at time $t_u$, we set the exposures at all future time steps on path $i$ to zero.

	In the BSM and MJD models, we estimate early-exercise boundaries for each $t_u$ by the bisection method, comparing the continuation value at $t_u$ with the payoff at $t_u$.
	
	In the HSV model, we first divide the interpolation domain in the variance dimension  $[\underline{\nu}, \overline{\nu}]$ at $t_u$ in 10 bins. To obtain the early exercise boundary at $t_u$ in the full re-evaluation approach, we again apply the bisection method, as in the BSM and MJD models, for each bin.
	
	In the accelerated HSV case, for each variance bin, we interpolate the early-exercise boundary in time using $5$ Chebyshev points on $[0,0.9]$, and use the bisection method for $t_u \in [0.9, 1)$. This way we achieve additional run-time savings, as instead of performing bisection method for each of 52 time steps $t_u$, we only perform it $10$ times, in the 5 Chebyshev points on $[0,0.9]$, and for the last five time steps of the exposure simulation time-grid. 
	
	We note that at this point we have obtained path-wise exposures of both approaches, full re-evaluation and Chebyshev acceleration,
	\begin{equation*}
	\left\{ \mathbf{x}_{t_u} = \left( x_{t_u}^1, ..., x_{t_u}^n \right) \right\}_{u=1}^m \hspace{0.5cm}\text{and}\hspace{0.5cm} \left\{\mathbf{y}_{t_u} = \left( y_{t_u}^1, ..., y_{t_u}^n \right)\right\}_{u=1}^m,
	\end{equation*}
respectively. From path-wise exposures, we can extract empirical estimators $\widehat{\rho}$ of exposure measure $\rho$ at time points of interest. Particularly, we will consider $\mathrm{EE}$, $\mathrm{PFE}_{0.95}$, and $\mathrm{CES}_{0.95}$, with explicit formulas for their empirical estimators, which are given in Appendix \ref{sec:exp_mea_exa}.

\emph{5. Calculate empirical time-profiles of the 3 considered exposure measures via full re-evalutation, $\left\{ \widehat{\rho}(\mathbf{x}_{t_u})\right\}_{u=1}^m$, and the accelerated approach, $\left\{ \widehat{\rho}(\mathbf{y}_{t_u})\right\}_{u=1}^m$}.

\emph{6. Calculate the maximum relative error in the exposure measure's time-profile induced by the function approximation,}
	\begin{equation}\label{eq:mre}
		\varepsilon^{\mathrm{CR}}_{\widehat{\rho}} = \max_{u=1,...,m} \frac{\left| \widehat{\rho}(\mathbf{x}_{t_u}) - \widehat{\rho}(\mathbf{y}_{t_u})\right|}{\widehat{\rho}(\mathbf{x}_{t_u})}.
		\end{equation}

		In the experiments of Section \ref{subsec:num_res} we chose to base the target acceleration accuracy on the statistical accuracy of the standard simulation-based methodology. For this, we quantify the MC errors as follows.

\emph{7. Asses the statistical accuracy of the underlying MC simulation.}
	
	Denote with $u_{\widehat{\rho}}^*$ the time index where the maximum in \eqref{eq:mre} is achieved for the exposure measure $\widehat{\rho}$. We first calculate the lengths of the 95\% confidence intervals for the MC estimates $\widehat{\rho}\left(\mathbf{x}_{t_{u_{\widehat{\rho}}^*}}\right)$,
	\begin{equation*}
	 l\left( \widehat{\rho}\left(\mathbf{x}_{t_{u_{\widehat{\rho}}^*}}\right) \right) = 2 q_{\frac{\alpha}{2}}\frac{\sigma\left( \widehat{\rho}\left(\mathbf{x}_{t_{u_{\widehat{\rho}}^*}}\right) \right)}{\sqrt{n}}
	\end{equation*}
	where $q_{\frac{\alpha}{2}}=\Phi^{-1}(1-\frac{\alpha}{2})$, with $\alpha = 0.05$, and $\sigma^2\left( \widehat{\rho}\left(\cdot\right) \right)$ is the MC sample estimate of the variance in an appropriate asymptotic result, \eqref{eq:EE_clt}, \eqref{eq:PFE_clt}, or \eqref{eq:CES_clt}. Then, we compute the relative error of the MC simulation for the exposure measure $\widehat{\rho}$ as
	\begin{equation*}
	 \varepsilon^{\mathrm{MC}}_{\widehat{\rho}} = \frac{l\left( \widehat{\rho}\left(\mathbf{x}_{t_{u_{\widehat{\rho}}^*}}\right) \right)}{\widehat{\rho}\left(\mathbf{x}_{t_{u_{\widehat{\rho}}^*}}\right)}.
	\end{equation*}
	
\emph{8. Calculate speed-up factors, i.e., factors by which the accelerated
approach is faster than a standard full re-evaluation.}
		
		For all options except American, speed-up factors are calculated as
		\begin{equation*}
		\frac{\text{time needed for Step 2}}{\text{time needed for Step 3}}.
		\end{equation*}
		
		For American options, we additionally account for the times needed to determine early-exercise policies, and so
		\begin{equation*}
		\frac{\text{time needed for Step 2 + time needed for Step 4}}{\text{time needed for Step 3  + time needed for Step 4}}.
		\end{equation*}
	
	The numerical experiments of Sections \ref{subsec:num_res} and \ref{subsec:ad_ch}  were conducted on the high performance computing cluster at Queen Mary University of London, Apocrita \cite{king_2017_438045}. Specifically, the experiments utilized 4 cores of the Intel(R) Xeon(R) Platinum 8268 CPU @ 2.90GHz processor, with 2GB RAM allocated per core. The code was implemented in Python 3.7.6. The numerical experiments of Section \ref{sec:cheby_greeks} were conducted on a personal computer equipped with an Intel(R) Core(TM) i7-9850H CPU @ 2.60GHz and 16GB RAM. The code was implemented in Python 3.7.7.

{\subsection{Theoretical foundation of the domain splitting}
Corollary~\ref{cor:cheb_conv} gives us insights in which cases a domain decomposition of a Chebyhev interpolation may be beneficial. As the construction of a Chebyshev interpolation of degree $N$ requires $N+1$  function evaluations, we can split the domain and construct two Chebyshev interpolations of half the degree for basically the same number of function evaluations.  
For differentiantiablity of order $p$, the corollary shows that the error in one dimension is bounded by a constant times
\[
\frac{L^p}{(N-p)^p},
\]
where $L$ is the length of the interpolation domain. Taking half the domain size and half the degree yields an estimate of
\[
\frac{L^p}{(N-2p)^p},
\]
which for large $N$ is asymptotically the same as the original estimate. The smaller $p$, the closer both estimates are, indicating that a domain split is justified where a function has low smoothness. In our case, this is at the strike.

For analytic functions a domain decomposition is not advisable for high enough degrees. We can follow Assumption~\eqref{ass2} for $\alpha=2$, yielding an interpolation error in one dimension of a constant times
\[
e^{b(L^2-N)}.
\]
Again keeping the effort equal, we can consider two domains of half the size and half the interpolation degree. Then the error bound is given as
\[
e^{\frac{b}{2}(\frac12 L^2-N)}
\]
For $N> L^2$, this is a larger estimate than the original estimate, indicating that for analytic functions increasing the degree is more efficient than a domain splitting.

}
\bibliographystyle{abbrv}  
\bibliography{refs} 

\begin{thebibliography}{10}

\bibitem{crepey2018}
L.~Abbas-Turki, S.~Cr\'{e}pey, and B.~Diallo.
\newblock {XVA principles, nested Monte Carlo strategies, and GPU
  optimizations}.
\newblock {\em International Journal of Theoretical and Applied Finance},
  21(06):1850030, 2018.

\bibitem{ACERBI2002srm}
C.~Acerbi.
\newblock Spectral measures of risk: A coherent representation of subjective
  risk aversion.
\newblock {\em Journal of Banking \& Finance}, 26(7):1505--1518, 2002.

\bibitem{acerbi2002}
C.~Acerbi and D.~Tasche.
\newblock {Expected shortfall: A natural coherent alternative to value at
  risk}.
\newblock {\em Economic Notes}, 31(2):379--388, 2002.

\bibitem{Andersen2017IM}
L.~Andersen, M.~Pykhtin, and A.~Sokol.
\newblock Credit exposure in the presence of initial margin.
\newblock {\em SSRN Electronic Journal}, 2017.

\bibitem{Andersen2017MPOR}
L.~Andersen, M.~Pykhtin, and A.~Sokol.
\newblock Rethinking the margin period of risk.
\newblock {\em Journal of Credit Risk}, 13(1):1--45, 2017.

\bibitem{andersson:25}
K.~Andersson and A.~Gnoatto.
\newblock Multi-layer deep {xVA}: Structural credit models, measure changes and
  convergence analysis.
\newblock {\em preprint}, 2025.
\newblock \url{https://arxiv.org/pdf/2502.14766}.

\bibitem{artzner1999}
P.~Artzner, F.~Delbaen, J.-M. Eber, and D.~Heath.
\newblock Coherent measures of risk.
\newblock {\em Mathematical Finance}, 9(3):203--228, 1999.

\bibitem{bcbs2014frtb}
{Basel Committee on Banking Supervision}.
\newblock {\em Fundamental review of the trading book: A revised market risk
  framework}.
\newblock Bank for International Settlements, 2014.

\bibitem{bcbs2020}
{Basel Committee on Banking Supervision}.
\newblock {\em {Margin requirements for non-centrally cleared derivatives}}.
\newblock Bank for International Settlements, 2020.

\bibitem{bayraktar2009}
E.~Bayraktar.
\newblock {A proof of the smoothness of the finite time horizon American put
  option for jump diffusions}.
\newblock {\em SIAM Journal on Control and Optimization}, 48(2):551--572, 2009.

\bibitem{Black1973}
F.~Black and M.~Scholes.
\newblock {The pricing of options and corporate liabilities}.
\newblock {\em Journal of Political Economy}, 81(3):637--654, 1973.

\bibitem{bo2014}
L.~Bo and A.~Capponi.
\newblock Bilateral credit valuation adjustment for large credit derivatives
  portfolios.
\newblock {\em Finance and Stochastics}, 18:431–482, 2014.

\bibitem{broadie2011}
M.~Broadie, Y.~Du, and C.~C. Moallemi.
\newblock Efficient risk estimation via nested sequential simulation.
\newblock {\em Management Science}, 57(6):1172--1194, 2011.

\bibitem{broadie2015}
M.~Broadie, Y.~Du, and C.~C. Moallemi.
\newblock Risk estimation via regression.
\newblock {\em Operations Research}, 63(5):1077--1097, 2015.

\bibitem{Canuto1982ApproximationRF}
C.~Canuto and A.~M. Quarteroni.
\newblock {Approximation results for orthogonal polynomials in Sobolev spaces}.
\newblock {\em Mathematics of Computation}, 38:67--86, 1982.

\bibitem{cont2010}
R.~Cont, R.~Deguest, and G.~Scandolo.
\newblock Robustness and sensitivity analysis of risk measurement procedures.
\newblock {\em Quantitative Finance}, 10(6):593--606, 2010.

\bibitem{crepey2016}
S.~Cr{\'e}pey and S.~Song.
\newblock Counterparty risk and funding: Immersion and beyond.
\newblock {\em Finance and Stochastics}, 20(4):901--930, Oct 2016.

\bibitem{crepey2015a}
S.~Crépey.
\newblock {Bilateral counterparty risk under funding constraints -- part I:
  Pricing}.
\newblock {\em Mathematical Finance}, 25(1):1--22, 2015.

\bibitem{crepey2015b}
S.~Crépey.
\newblock {Bilateral counterparty risk under funding constraints -- part II:
  CVA}.
\newblock {\em Mathematical Finance}, 25(1):23--50, 2015.

\bibitem{fang2009euro}
F.~Fang and C.~W. Oosterlee.
\newblock {A novel pricing method for European options based on Fourier-cosine
  series expansions}.
\newblock {\em SIAM Journal on Scientific Computing}, 31(2):826--848, 2009.

\bibitem{fang2009}
F.~Fang and C.~W. Oosterlee.
\newblock {Pricing early-exercise and discrete barrier options by
  Fourier-cosine series expansions}.
\newblock {\em Numerische Mathematik}, 114(1), 2009.

\bibitem{follmer2002}
H.~F{\"o}llmer and A.~Schied.
\newblock Convex measures of risk and trading constraints.
\newblock {\em Finance and Stochastics}, 6(4):429--447, Oct. 2002.

\bibitem{follmer2008}
H.~F\"{o}llmer and A.~Schied.
\newblock {\em Stochastic Finance: An Introduction in Discrete Time}.
\newblock De Gruyter, 2016.

\bibitem{fox1968chebyshev}
L.~Fox and I.~Parker.
\newblock {\em Chebyshev Polynomials in Numerical Analysis}.
\newblock Oxford University Press, 1968.

\bibitem{gass2016}
M.~Ga{\ss}.
\newblock {PIDE Methods and Concepts for Parametric Option Pricing}.
\newblock {\em Ph.D. thesis, Technical University of Munich}, 2016.

\bibitem{gass2018}
M.~Ga{\ss}, K.~Glau, M.~Mahlstedt, and M.~Mair.
\newblock Chebyshev interpolation for parametric option pricing.
\newblock {\em Finance and Stochastics}, 22(3):701--731, 2018.

\bibitem{gentleman1972}
W.~M. Gentleman.
\newblock {Implementing Clenshaw-Curtis quadrature, II computing the cosine
  transformation}.
\newblock {\em Commun. ACM}, 15(5):343–346, 1972.

\bibitem{giesecke:13}
K.~Giesecke, K.~Spiliopoulos, and R.~B. Sowers.
\newblock Default clustering in large portfolios: Typical events.
\newblock {\em Annals of Applied Probability}, 23(1):348--385, 2013.

\bibitem{giesecke:15}
K.~Giesecke, K.~Spiliopoulos, R.~B. Sowers, and J.~A. Sirignano.
\newblock Large portfolio asymptotics for loss from default.
\newblock {\em Mathematical Finance}, 25(1):77--114, 2015.

\bibitem{Gilli2012}
M.~Gilli and E.~Schumann.
\newblock Calibrating option pricing models with heuristics.
\newblock In A.~Brabazon, M.~O'Neill, and D.~Maringer, editors, {\em Natural
  Computing in Computational Finance: Volume 4}, pages 9--37. Springer Berlin
  Heidelberg, 2012.

\bibitem{glasserman2004monte}
P.~Glasserman.
\newblock {\em {Monte Carlo Methods in Financial Engineering}}.
\newblock {Stochastic Modelling and Applied Probability}. Springer, 2003.

\bibitem{Glasserman2000VarianceRT}
P.~Glasserman, P.~Heidelberger, and P.~Shahabuddin.
\newblock {Variance reduction techniques for estimating value-at-risk}.
\newblock {\em Management Science}, 46:1349--1364, 2000.

\bibitem{glau19iv}
K.~Glau, P.~Herold, D.~B. Madan, and C.~P\"{o}tz.
\newblock {The Chebyshev method for the implied volatility}.
\newblock {\em Journal of Computational Finance}, 23(3), 2019.

\bibitem{glau2020}
K.~Glau, D.~Kressner, and F.~Statti.
\newblock {Low-rank tensor approximation for Chebyshev interpolation in
  parametric option pricing}.
\newblock {\em SIAM Journal on Financial Mathematics}, 11(3):897--927, 2020.

\bibitem{glau2019improved}
K.~Glau and M.~Mahlstedt.
\newblock {Improved error bound for multivariate Chebyshev polynomial
  interpolation}.
\newblock {\em International Journal of Computer Mathematics},
  96(11):2302--2314, 2019.

\bibitem{glau19dc}
K.~Glau, M.~Mahlstedt, and C.~P\"{o}tz.
\newblock { A new approach for American option pricing: The Dynamic Chebyshev
  method}.
\newblock {\em SIAM Journal on Scientific Computing}, 41(1):B153--B180, 2019.

\bibitem{glau2019fast}
K.~Glau, R.~Pachon, and C.~Pötz.
\newblock Speed-up credit exposure calculations for pricing and risk
  management.
\newblock {\em Quantitative Finance}, 21(3):481--499, 2021.

\bibitem{wunderlich2023}
K.~Glau and L.~Wunderlich.
\newblock Neural network expression rates and applications of the deep
  parametric {PDE} method in counterparty credit risk.
\newblock {\em Annals of Operations Research}, Apr. 2023.

\bibitem{reisinger2023}
A.~Gnoatto, A.~Picarelli, and C.~Reisinger.
\newblock {Deep xVA solver: a neural network–based counterparty credit risk
  management framework}.
\newblock {\em SIAM Journal on Financial Mathematics}, 14(1):314--352, 2023.

\bibitem{Gordy2010}
M.~B. Gordy and S.~Juneja.
\newblock Nested simulation in portfolio risk measurement.
\newblock {\em Management Science}, 56(10):1833--1848, 2024/02/15/ 2010.

\bibitem{green2015}
A.~Green.
\newblock {\em XVA: Credit, Funding and Capital Valuation Adjustments}.
\newblock Wiley, 2015.

\bibitem{gregory2020xva}
J.~Gregory.
\newblock {\em {The xVA Challenge: Counterparty Risk, Funding, Collateral,
  Capital and Initial Margin}}.
\newblock Wiley Finance. Wiley, 2020.

\bibitem{GRZELAK2022}
L.~A. Grzelak.
\newblock {Sparse grid method for highly efficient computation of exposures for
  xVA}.
\newblock {\em Applied Mathematics and Computation}, 434:127446, 2022.

\bibitem{Herrmann2022}
L.~Herrmann, J.~A.~A. Opschoor, and C.~Schwab.
\newblock {Constructive Deep ReLU Neural Network Approximation}.
\newblock {\em Journal of Scientific Computing}, 90(2):75, 2022.

\bibitem{Heston1993}
S.~Heston.
\newblock {A closed-form solution for options with stochastic volatility with
  applications to bond and currency options}.
\newblock {\em Review of Financial Studies}, 6:327--343, 1993.

\bibitem{hout2008}
K.~J. {in 't Hout} and S.~{Foulon}.
\newblock {ADI finite difference schemes for option pricing in the Heston model
  with correlation}.
\newblock {\em International Journal of Numerical Analysis and Modeling}, 7, 11
  2008.

\bibitem{isda2021}
{International Swaps and Derivatives Association}.
\newblock {\em {ISDA Standard Initial Margin Model (ISDA SIMM) Methodology,
  version 2.4}}.
\newblock 2021.

\bibitem{jain2015}
S.~Jain and C.~W. Oosterlee.
\newblock {The stochastic grid bundling method: Efficient pricing of Bermudan
  options and their Greeks}.
\newblock {\em Applied Mathematics and Computation}, 269:412--431, 2015.

\bibitem{karlsson2016}
P.~Karlsson, S.~Jain, and C.~W. Oosterlee.
\newblock {Counterparty credit exposures for interest rate derivatives using
  the stochastic grid bundling method}.
\newblock {\em Applied Mathematical Finance}, 23(3):175--196, 2016.

\bibitem{kenyon2018}
C.~Kenyon, M.~Berrahoui, and B.~Poncet.
\newblock {Counterparty trading limits revisited: CSAs, IM, SwapAgent®, from
  PFE to PFL}.
\newblock {\em SSRN Electronic Journal}, 2018.

\bibitem{king_2017_438045}
T.~King, S.~Butcher, and L.~Zalewski.
\newblock {\em {Apocrita - High Performance Computing Cluster for Queen Mary
  University of London}}, 2017.

\bibitem{Longstaff}
F.~Longstaff and E.~Schwartz.
\newblock {Valuing American options by simulation: A simple least-squares
  approach}.
\newblock {\em Review of Financial Studies}, 14:113--47, 2001.

\bibitem{marshall2011}
A.~W. Marshall, I.~Olkin, and B.~C. Arnold.
\newblock {\em Inequalities: Theory of Majorization and Its Applications}.
\newblock Springer, second edition edition, 2011.

\bibitem{Merton1973}
R.~C. Merton.
\newblock {Theory of rational option pricing}.
\newblock {\em The Bell Journal of Economics and Management Science},
  4(1):141--183, 1973.

\bibitem{Merton1976}
R.~C. Merton.
\newblock Option pricing when underlying stock returns are discontinuous.
\newblock {\em Journal of Financial Economics}, 3(1):125 -- 144, 1976.

\bibitem{nasdaq2024}
{Nasdaq}.
\newblock {Nasdaq Integrates AI to Simplify and Accelerate Bank and Insurance
  Risk Calculations}.
\newblock {\em {Press release}}, 2024.

\bibitem{feng2016}
C.~W. Oosterlee, Q.~Feng, S.~Jain, P.~Karlsson, and D.~Kandhai.
\newblock {Efficient computation of exposure profiles on real-world and
  risk-neutral scenarios for Bermudan swaptions}.
\newblock {\em Journal of Computational Finance}, 20(1):139--172, 2016.

\bibitem{pachon2018}
R.~Pach\'{o}n.
\newblock {Numerical pricing of European options with arbitrary payoffs}.
\newblock {\em International Journal of Financial Engineering}, 05(02):1850015,
  2018.

\bibitem{Peskir2006OptimalSA}
G.~Peskir and A.~Shiryaev.
\newblock {\em Optimal Stopping and Free-Boundary Problems}.
\newblock Springer, New York, 2006.

\bibitem{potz2020}
C.~P\"{o}tz.
\newblock {Function approximation for option pricing and risk management:
  Methods, theory and applications}.
\newblock {\em Ph.D. thesis, Queen Mary University of London}, 2020.

\bibitem{Rachev2013TheMO}
S.~T. Rachev, L.~B. Klebanov, S.~Stoyanov, and F.~J. Fabozzi.
\newblock {\em {The Methods of Distances in the Theory of Probability and
  Statistics}}.
\newblock Springer, New York, 2013.

\bibitem{renzitti2020}
S.~Renzitti, P.~Bastani, and S.~Sivorot.
\newblock {Accelerating CVA and CVA sensitivities using quasi-Monte Carlo
  methods}.
\newblock {\em Wilmott}, 2020(108):78--93, 2020.

\bibitem{sauter2011}
S.~Sauter and C.~Schwab.
\newblock {\em Boundary Element Methods}, volume~39 of {\em Springer Series
  Computational Mathematics.}
\newblock Springer, 2011.

\bibitem{Schftner2008OnTE}
R.~Sch{\"o}ftner.
\newblock {On the estimation of credit exposures using regression-based Monte
  Carlo simulation}.
\newblock {\em Journal of Credit Risk}, 4:37--62, 2008.

\bibitem{sterfling1980}
R.~Serfling.
\newblock {\em {Approximation Theorems of Mathematical Statistics}}.
\newblock {Wiley Series in Probability and Mathematical Statistics}. Wiley,
  1980.

\bibitem{Shen2013ABA}
Y.~Shen, J.~A.~M. van~der Weide, and J.~H.~M. Anderluh.
\newblock {A benchmark approach of counterparty credit exposure of Bermudan
  option under L{\'e}vy process: the Monte Carlo-COS method}.
\newblock {\em Procedia Computer Science}, 18:1163--1171, 2013.

\bibitem{Storn1997}
R.~Storn and K.~Price.
\newblock {Differential evolution -- A simple and efficient heuristic for
  global optimization over continuous spaces}.
\newblock {\em Journal of Global Optimization}, 11(4):341--359, 1997.

\bibitem{quantlib}
{The QuantLib contributors}.
\newblock {QuantLib: a free/open-source library for quantitative finance}.
\newblock \url{https://www.quantlib.org/}.

\bibitem{trefethen2013}
L.~N. Trefethen.
\newblock {\em {Approximation Theory and Approximation Practice}}.
\newblock Society for Industrial and Applied Mathematics, 2013.

\bibitem{trinidade2007}
A.~A. Trindade, S.~Uryasev, A.~Shapiro, and G.~Zrazhevsky.
\newblock Financial prediction with constrained tail risk.
\newblock {\em Journal of Banking and Finance}, page 3524–3538, 2007.

\bibitem{weber2004}
S.~Weber.
\newblock {\em {Measures and Models of Financial Risk}}.
\newblock {Ph.D. Thesis}, Humboldt-Universit{\"a}t zu Berlin, 2004.

\bibitem{weber2006}
S.~Weber.
\newblock Distribution-invariant risk measures, information, and dynamic
  consistency.
\newblock {\em Mathematical Finance}, 16(2):419--441, 2006.

\bibitem{wilkes2023}
S.~Wilkes.
\newblock Party’s over as more banks drop internal models for market risk.
\newblock {\em Risk.net}, 2023.

\bibitem{zeron2021}
M.~Zeron and I.~Ruiz.
\newblock {Tensoring dynamic sensitivities and dynamic initial margin}.
\newblock {\em Risk Magazine}, 2021.

\end{thebibliography}

\end{document}